\pdfoutput=1
\documentclass[conference]{IEEEtran}
\IEEEoverridecommandlockouts
\usepackage{cite}
\usepackage[fleqn]{amsmath}
\usepackage{amsthm}
\usepackage{amssymb}
\usepackage{amsfonts}
\usepackage{algorithmic}
\usepackage{graphicx}
\usepackage{textcomp}
\usepackage[numbers]{natbib}
\usepackage{url}
\usepackage{mathpartir}
\usepackage{mathrsfs} 
\usepackage{semantic} 
\usepackage{color}
\usepackage{mathtools}  
\usepackage{thmtools}
\usepackage{xspace}
\usepackage{xcolor}
\usepackage{hyperref}
\usepackage{anyfontsize}
\usepackage{amsbsy}       
\usepackage{stmaryrd}
\usepackage{braket}
\usepackage{marvosym} 
\usepackage{wasysym}    
\usepackage{balance}    
\usepackage{proof}
\usepackage{yhmath}
\usepackage{rotating}
\usepackage{tikz-cd}
\usepackage{changepage}
\usepackage{centernot}
\usepackage{comment}
\usepackage{listings}
\usepackage{mdframed}
\usepackage{centernot}
\usepackage{xparse}
\usepackage{upgreek}
\usepackage{inconsolata} 

\usepackage{multicol}

\def\BibTeX{{\rm B\kern-.05em{\sc i\kern-.025em b}\kern-.08em
    T\kern-.1667em\lower.7ex\hbox{E}\kern-.125emX}}
		
\newcommand{\icode}[1]{$\mathsf{#1}$}

\newenvironment{block}[1][t]
  {\begin{array}[#1]{@{}l@{}}}
  {\end{array}}

\definecolor{lightgray}{gray}{0.90}
\newcommand{\Gbox}[1]{\colorbox{lightgray}{$#1$}}

\newtheorem{theorem}{Theorem}
\newtheorem{definition}{Definition}
\newtheorem{lemma}[theorem]{Lemma}

\declaretheorem[numbered=no]{case}

\mathlig{[]}{\square}
\mathlig{|-s}{\vdash^{\!\forall}}
\mathlig{|-}{\vdash}
\mathlig{|}{\mid} 

\reservestyle{\oblang}{\mathsf}
\oblang{if[if\;],then[\;then\;],else[\;else\;],let[let\;],in[\;in\;]}

\newcommand{\Int}{\mathsf{Int}}
\newcommand{\String}{\mathsf{String}}

\newcommand{\Bool}{\mathsf{Bool}}
\newcommand{\Unit}{\mathsf{Unit}}
\newcommand{\unit}{\mathsf{unit}}

\newcommand{\StringEq}{\mathsf{StringEq}}
\newcommand{\StringHashEq}{\mathsf{StringHashEq}}
\newcommand{\IntEq}{\mathsf{IntEq}}
\newcommand{\StringLen}{\mathsf{StringLen}}

\newcommand{\dom}{\mathit{dom}}

\newcommand{\sub}{<:}
\newcommand{\join}{\sqcup}

\newcommand{\ie}{\emph{i.e.}\xspace}
\newcommand{\etal}{\emph{et al.}\xspace}
\newcommand{\eg}{\emph{e.g.}\xspace}

\mathlig{++}{\mathbin{++}}

\newcommand{\ltop}{\textsf{H}}
\newcommand{\lbot}{\textsf{L}}

\newcommand\obsec[0]{$\mathsf{Ob}_{\mathsf{SEC}}$\xspace}

\newcommand\defas[0]{\stackrel{\triangle}{=}}

\newcommand{\object}{\left[z : S => \overline{\metdef} \right]}
\newcommand{\objectx}[2]{\left[#1 : #2 => \overline{\metdef} \right]}
\newcommand{\objectxx}[3]{\left[#1 : #2 => #3 \right]}

\newcommand{\stypetop}[1]{#1_{\ltop}}
\newcommand{\stypebot}[1]{#1_{\lbot}}

\newcommand{\metdef}{m\left(x\right)e }

\newcommand{\methinvx}[2]{#1.m(#2)}

\newcommand{\ntrecordx}[3]{\left[\overline{#1 : #2 \rightarrow #3}\right]}
\newcommand{\ntrecord}{\ntrecordx{m}{S}{S}}

\newcommand{\rtypex}[2]{\textbf{Obj}(#1) . #2}
\newcommand{\stype}[2]{{#1\triangleleft#2}}
\newcommand{\rtop}{\left[ \hspace*{0.06cm} \right]}

\newcommand\ssubst[3]{#1\left[#2 / #3\right]}
\newcommand\mbsubst[5]{\ssubst{\ssubst{#1}{#2}{#3}}{#4}{#5}}

\newcommand{\reduce}{\longmapsto}

\newcommand{\GammaSub}[0]{\Phi}
\newcommand{\GammaST}[0]{\Sigma}

\newcommand{\unfold}[1]{\GammaST\left[#1\right]}

\newcommand\extgamma[3]{#1 \left[#2 \mapsto #3 \right]}
\newcommand\extgammax[5]{\extgamma{\extgamma{#1}{#2}{#3}}{#4}{#5}}

\newcommand{\typeeq}[2]{#1 \equiv #2}
\newcommand{\stypeeq}[2]{#1 \equiv #2}

\newcommand{\stypeof}[3]{#1 \vdash_{1} #2 : #3}
\newcommand{\subtp}[3]{#1 \vdash #2 \sub #3}
\newcommand{\subtps}[3]{\subtp{#1}{#2}{#3}}

\newcommand\tlookup[3]{\mathsf{msig}(#1,#2)=#3}
\newcommand\methimpl[3]{\mathsf{methimpl}(#1,#2)=#3}

\newcommand\rni[3]{\mathsf{TRNI}(#1,#2,#3)}

\newcommand\safevdash[0]{\vdash_{\mathsf{sf}}}
\newcommand\simplesub[3]{#1 \vdash_{\mathsf{sf}} #2 <: #3}

\newcommand{\setvx}[2]{\mathcal{V}_{#1}\llbracket#2\rrbracket}
\newcommand{\setcx}[2]{\mathcal{C}_{#1}\llbracket#2\rrbracket}

\definecolor{lightred}{RGB}{255,100,100}

\newcommand{\fsub}{\mathsf{F}_{<:}}

\newcommand\gobsec[0]{$\mathsf{Ob}^{\left\langle \right\rangle}_{\mathsf{SEC}}$\xspace}
\newcommand{\gntrecordx}[2]{\left[\overline{#1 : #2}\right]}

\newcommand{\gobject}{\left[z : S => \overline{\metdef} \right]}
\newcommand{\gmtype}[3]{\left\langle #1 \right\rangle #2 \rightarrow #3}

\newcommand{\mtype}[2]{#1 \rightarrow #2}

\newcommand{\gminv}[4]{#1.#2\left\langle #3\right\rangle(#4)}
\newcommand{\minv}[3]{#1.#2(#3)}

\newcommand{\StringFirst}{\mathsf{StringFst}}
\newcommand{\ListStrFstLen}{\mathsf{ListStrFstLen}}
\newcommand{\StrFstLen}{\mathsf{StrFstLen}}
\newcommand{\ListStrLen}{\mathsf{ListStrLen}}

\newcommand{\StringEqL}{\mathsf{StringEqL}}
\newcommand{\StringEqBad}{\mathsf{StringEqBad}}
\newcommand{\StringEqPoly}{\mathsf{StringEqPoly}}
\newcommand{\ListEqStr}{\mathsf{ListEqStr}}
\newcommand{\ListStr}{\mathsf{ListStr}}

\newcommand{\ilab}{*}
\newcommand{\rpolicy}{\mathsf{rdecl}}

\newcommand{\DeltaX}{\Delta}
\newcommand{\GammaTOk}[0]{\Delta_{\mathsf{ok}}}
\newcommand\wfe[1]{\vdash #1}
\newcommand\wfex[2]{#1 \vdash #2}
\newcommand\wfft[2]{#1 \vdash_{\triangleleft} #2}
\newcommand\wft[2]{#1 \vdash #2}
\newcommand\wf[2]{#1 \models #2}

\newcommand{\gtlookup}[3]{\mathsf{msig}(#1,#2,#3)}

\newcommand{\goinst}[2]{#1\left\langle #2 \right\rangle}

\newcommand{\primt}{P}
\newcommand{\primb}{{\mathsf{p}}}

\newcommand{\methods}{\mathsf{meths}}

\newcommand{\ptsound}[1]{\models #1}

\newcommand{\atomone}[2]{{\text{Atom}_{#1}\left[#2\right]}}

\newcommand{\atomvalone}[2]{{\text{Atom}^{val}_{#1}\left[#2\right]}}
\newcommand{\atomunion}[1]{{\text{Atom}\left[#1\right]}}  

\newcommand{\downreln}[1]{\left\lfloor R\right\rfloor_{n}}

\newcommand{\rhosyn}[1]{\rho(#1)} 
\newcommand{\rhosynx}[2]{{#1}_\mathsf{syn}(#2)}

\newcommand{\gsetv}[1]{\gsetvx{#1}{\rho}}
\newcommand{\gsetvx}[2]{{\mathcal{V}\llbracket#1\rrbracket}{#2}}

\newcommand{\gsetc}[1]{\gsetcx{#1}{\rho}}
\newcommand{\gsetcx}[2]{{\mathcal{C}\llbracket#1\rrbracket}{#2}}

\newcommand{\gsetg}[1]{\gsetgx{#1}{\rho}}
\newcommand{\gsetgx}[2]{\mathcal{G}\llbracket#1\rrbracket{#2}}

\newcommand{\gsetd}[1]{\mathcal{D}\llbracket#1\rrbracket}

\newcommand{\gtrni}[4]{\mathsf{PRNI}(#1,#2,#3,#4)}

\newcommand{\nsetv}[1]{\nsetvx{#1}}
\newcommand{\nsetvx}[1]{{\mathcal{V}\llbracket#1\rrbracket}}
\newcommand{\nsetc}[1]{\nsetcx{#1}}
\newcommand{\nsetcx}[1]{{\mathcal{C}\llbracket#1\rrbracket}}
\newcommand{\nsetg}[1]{\nsetgx{#1}}
\newcommand{\nsetgx}[1]{\mathcal{G}\llbracket#1\rrbracket}
\newcommand{\nsetd}[1]{\mathcal{D}\llbracket#1\rrbracket}
\newcommand{\xsigma}[2]{\sigma\left[#1 \mapsto #2 \right]}

\newcommand{\ubound}[2]{\mathsf{ub}(#1,#2)}

\newcommand{\inhole}[2]{#1\left[#2\right]}

\makeatletter
\newcommand*\bigcdot{\mathpalette\bigcdot@{.7}}
\newcommand*\bigcdot@[2]{\mathbin{\vcenter{\hbox{\scalebox{#2}{$\m@th#1\bullet$}}}}}
\makeatother

\newcommand{\wfsub}[3]{#1 \vdash #2 \blacktriangleleft #3}

\definecolor{darkgreen}{RGB}{0,128,0}
\lstdefinelanguage{scala}{
    morekeywords={let,abstract,case,catch,class,def,%
      do,else,extends,false,final,finally,%
      for,if,implicit,import,match,mixin,%
      new,null,object,override,package,%
      private,protected,requires,return,sealed,%
      super,this,throw,trait,true,try,%
      type,val,var,while,with,yield, app, has,
			top,bottom,declassify,Obj},
    sensitive=true,
		keywordstyle={\color{blue}},
    morecomment=[l][\color{darkgreen}]{//},
    morecomment=[n]{/*}{*/},
    morestring=[b]",
    morestring=[b]',
    morestring=[b]""",
    escapeinside={(*}{*)},
    moredelim=**[is][{\btHL}]{`}{`}
  }
\begin{document}

\title{Polymorphic Relaxed Noninterference}

\author{\IEEEauthorblockN{Raimil Cruz\thanks{This work is partially funded by CONICYT FONDECYT Regular Projects 1150017 and 1190058, and is in part supported by the \href{https://erc.europa.eu/}{European Research Council} under \href{https://secure-compilation.github.io/}{ERC Starting Grant SECOMP (715753).}
Raimil Cruz is partially funded by CONICYT-PCHA/Doctorado Nacional/2014-63140148}}
\IEEEauthorblockA{\textit{PLEIAD Lab, Computer Science Department (DCC)} \\
\textit{University of Chile}\\
Santiago, Chile \\
racruz@dcc.uchile.cl}
\and
\IEEEauthorblockN{\'Eric Tanter}
\IEEEauthorblockA{\textit{PLEIAD Lab, Computer Science Department (DCC)} \\
\textit{University of Chile},
Santiago, Chile \\
\& Inria Paris, France \\
etanter@dcc.uchile.cl}
}

\maketitle

\begin{abstract}
Information-flow security typing statically preserves confidentiality by enforcing noninterference. To address the practical need of selective and flexible declassification of confidential information, 
several approaches have developed a notion of {\em relaxed} noninterference, where security labels are either functions or types. The labels-as-types approach to relaxed noninterference supports expressive declassification policies, including recursive ones, via a simple subtyping-based ordering, and provides a local, modular reasoning principle.
In this work, we extend this expressive declassification approach in order to support {\em polymorphic} declassification. First, we identify the need for bounded polymorphism through concrete examples. We then formalize polymorphic relaxed noninterference in a typed object-oriented calculus, using a step-indexed logical relation to prove that all well-typed terms are secure. Finally, we address the case of primitive types, which requires a form of ad-hoc polymorphism. Therefore, this work addresses practical hurdles to
providing controlled and expressive declassification for the construction of 
information-flow secure systems.
\end{abstract}

\section{Introduction}
\label{sec:gobsec-introduction}

An information-flow security type system statically ensures that public outputs (\eg~$\String_{\lbot}$) cannot depend on secret inputs (\eg~$\String_{\ltop}$), a property known as noninterference (NI)~\cite{volpanoAl:jcs1996}.
NI provides a modular reasoning principle about security, indexed by the observational power of an adversary. 
For instance,  a function $f: \String_{\ltop} \rightarrow \String_{\lbot}$ does not reveal any information about its argument; in fact, 
in a pure language, it is necessarily a constant function.

But noninterference is too strict in practice: for a system to be useful, confidential information sometimes needs to be {\em declassified}. Beyond introducing a declassification operator in the language, which compromises formal reasoning, various approaches have explored structured ways to support declassification policies~\cite{sabelfeldSands:jcs2009,liZdancewic:popl2005,hicksAl:plas2006,cruzAl:ecoop2017}. In particular, \citet{liZdancewic:popl2005}
introduce {\em relaxed} noninterference, supporting expressive declassification policies via {\em security labels as functions}. 
In this approach, instead of having security labels such as $\ltop$ for private and $\lbot$ for public information that are drawn from a fixed lattice of symbols, security labels {\em are} the very functions that describe how a given secret can be manipulated in order to produce a public value: for instance, one can realize the declassification policy {\em ``only the result of comparing the hash of the secret string $s$ with a public guess can be made public''} by attaching to $s$ a function that implements this declassification ($\lambda x.\lambda y.~\mathsf{hash}(x)=y$). Any use of the secret that does not follow the declassification policy yields private results. 
One can express the standard label $\ltop$ (resp. $\lbot$) as a constant function (resp. the identity function). A challenging aspect of this approach is that label ordering relies on a semantic interpretation of declassification functions. 

A more practical approach than this {\em labels-as-functions} approach was recently developed by \citet{cruzAl:ecoop2017} in an object-oriented setting, with a {\em labels-as-types} perspective: security types are \emph{faceted types} of the form $\stype{T}{U}$ where the first facet $T$---called the {\em safety type}---represents the implementation type, exposed to the private observer, and the 
second facet $U$---called the {\em declassification type}--represents the declassification policy as an object interface exposed to the public observer.\footnote{To account for $k > 2$ observation levels, faceted types can be extended to have $k$ facets~\cite{cruzAl:ecoop2017}. Here, we restrict the presentation to two observation levels.}
For instance, the type
$\stype{\String}{\top}$, where $\top$ is the empty object interface, denotes 
private values (no method is declassified) and the type $\stype{\String}{\String}$ represents public values (all methods are declassified).
These security types are abbreviated as $\String_{\ltop}$ and $\String_{\lbot}$, respectively. Interesting declassification policies stand in between these two extremes: for instance, given the interface $\StringLen \triangleq \left[\mathsf{length}: \Unit_{\lbot} -> \Int_{\lbot}\right]$,
the faceted type $\stype{\String}{\StringLen}$ exposes the method $\mathsf{length}$ to declassify the length of a string as a public integer, but not its content.
This type-based approach to declassification 
 is expressive as well as simple---in particular, because labels are types, label ordering is simply subtyping. Also, it extends the modular reasoning principle of NI to account for declassification~\cite{cruzAl:ecoop2017}, a property named \emph{type-based relaxed noninterference} (TRNI). For instance, with TRNI one can prove that
a function of type $\stype{\String}{\StringLen} \rightarrow \Bool_{\lbot}$ 
{\em must} produce equal results for strings of equal lengths.

The labels-as-types approach of Cruz \etal however lacks {\em security label polymorphism}. 
Security label polymorphism is a very useful feature of practical security-typed languages such as JIF~\cite{myers:jif} and FlowCaml~\cite{pottierSimonet:toplas2003}, which has only been explored in the context of standard security labels (symbols from a lattice). To the best of our knowledge, polymorphism has not been studied for 
expressive declassification mechanisms, such as labels-as-functions~\cite{liZdancewic:popl2005} or labels-as-types~\cite{cruzAl:ecoop2017}. 
We extend the labels-as-types approach
with declassification {\em polymorphism}, specifically 
{\em bounded} polymorphism that specifies both a lower and an upper bound for a polymorphic declassification type.

The main contribution of this paper is to develop the theory of bounded
polymorphic declassification as an extension of TRNI, called {\em polymorphic relaxed noninterference} (PRNI for short). PRNI brings new benefits in the expressiveness and design of declassification interfaces. Additionally, we address the necessary support for primitive types, through a form of ad-hoc polymorphism.

The labels-as-types approach has the practical benefits of relying on concepts that are well-known to developers---object interfaces and subtyping---in order to build systems with information flow security that cleanly account for controlled and expressive declassification. This work addresses the two major shortcomings of prior work in order to bring this approach closer to real-world secure programming.

Section~\ref{sec:background} provides background on labels-as-types and TRNI. Section~\ref{sec:gobsec-overview} then explains the main aspects of polymorphic relaxed noninterference (PRNI).
Section~\ref{sec:gobsec-model} formalizes polymorphic declassification in a core object-oriented language \gobsec. Then, Section~\ref{sec:gobsec-trni} develops a logical relation for PRNI and shows that all well-typed \gobsec terms satisfy PRNI.
Section~\ref{sec:gobsec-primitive-types} extends \gobsec with primitive types.
Section~\ref{sec:gobsec-related-work} discusses related work and Section~\ref{sec:conclusion} concludes. 
We have implemented an interactive prototype of \gobsec that is available at \url{https://pleiad.cl/gobsec/}.

\section{Background:Type-Based Relaxed Noninterference}
\label{sec:background}

We start with a quick review of type-based relaxed noninterference~\cite{cruzAl:ecoop2017}. 
Faceted security types allow programmers to express declassification policies
as type interfaces. For instance, one can express that a \icode{login} function can reveal the result of comparing a secret password for equality with a public guess.

\begin{lstlisting}[numbers=none]
(*$\stypebot{\String}$*) login((*$\stypebot{\String}$*) guess, (*$\stype{\String}{\StringEq}$*) password){
 if(password.eq(guess))
    return "Login Successful"
  else
     return "Login failed"
}
\end{lstlisting}

Note that leaking the secret password by directly returning it would not typecheck, since $\StringEq$ is not a subtype of $\String$ (recall that $\stypebot{\String}$ is short for $\stype{\String}{\String}$).
Taking advantage of the fact that object types are recursive, one can also 
express \emph{recursive declassification}, for instance that a list of secret strings can only be declassified by comparing its elements for equality.
Likewise, one can express \emph{progressive declassification} 
by nesting type interfaces. For instance, assuming that $\String$ has a method $\mathsf{hash}: \stypebot{\Unit} \rightarrow \stypebot{\Int}$, we 
can specify that only the hash of the password can be compared for equality with the interface type 
$\StringHashEq \triangleq \left[\mathsf{hash}: \stypebot{\Unit} \rightarrow \stype{\Int}{\IntEq}\right]$,
where $\IntEq \triangleq \left[\mathsf{eq}: \stypebot{\Int} \rightarrow \stypebot{\Bool}\right]$:
\begin{lstlisting}[numbers=none]
(*$\stypebot{\String}$*) login((*$\stypebot{\Int}$*) guess, (*$\stype{\String}{\StringHashEq}$*) password){
  if(password.hash().eq(guess)) ...
}
\end{lstlisting}

Cruz~\etal formalize faceted security types in \obsec, a core object-oriented language with three 
kinds of expressions: variables, objects and method invocations (Figure~\ref{fig:obsec-syntax}). An object $\object$ is a collection of methods that can 
refer to the defining object with the self variable $z$. An object type $\rtypex{\alpha}{\ntrecord}$ is a collection of 
method signatures that have access to the defining type through the self type variable $\alpha$. Security types $S = \stype{T}{U}$ are composed of two object types $T$ and $U$.
Note that to be well-formed, a security type $\stype{T}{U}$ requires $U$ to be a {\em supertype} of $T$. The type abstraction mechanism of subtyping (by which a supertype ``hides'' members of its subtypes)
is the key element to express declassification.

\begin{figure}[t]
  \begin{small}		
		\begin{displaymath}
			\begin{array}{rcll}
				e & ::= & v |  e.m(e)| x & \text{(terms)}\\
				v & ::= & \object  & \text{(values)}\\					
				T, U & ::= & O | \alpha & \text{(types)}\\
				O & ::= & \rtypex{\alpha}{\ntrecord} & \text{(object types)}\\
				S & ::= & \stype{T}{U} & \text{(security types)}\\
			\end{array}		
		\end{displaymath}	
		\begin{minipage}{0.3\columnwidth}			
			\framebox{$\Gamma |- e : S$}
		\end{minipage}
		\begin{minipage}{0.7\columnwidth}
			$\Gamma ::= \bigcdot | \Gamma, x: S$ (type environment)
		\end{minipage}
		\begin{mathpar}
		\inference[(TmD)]{
				\Gamma |- e_1 : \stype{T}{U} &
				\Gbox{m \in U} \\
				\tlookup{U}{m}{S_1 \rightarrow S_2}	&
				\Gamma |- e_2 : S_1  
			}{
				\Gamma |- \methinvx{e_1}{e_2}: S_2
			}		
		\end{mathpar}
		\begin{mathpar}
			\inference[(TmH)]{
					\Gamma |- e_1 : \stype{T}{U} &
					\Gbox{m \notin U} \\
					\tlookup{T}{m}{S_1 \rightarrow \Gbox{\stype{T_2}{U_2}}} &
					\Gamma |- e_2 : S_1 
				}{
					\Gamma |- \methinvx{e_1}{e_2}: \stype{T_2}{\Gbox{\top}}
					}
		\end{mathpar}
 \end{small}
 \caption{\obsec: Syntax and Static semantics (excerpts from \cite{cruzAl:ecoop2017})}
  \label{fig:obsec-syntax}
\end{figure}

The \obsec type system defines two rules to give a type to a method invocation depending on whether the invoked method is in the declassification 
type or not. 
Rule (TmD) specifies that if the invoked method $m$ is in the declassification type $U$ with type $S_1 -> S_2$, then the result type of the method invocation expression is $S_2$. Conversely, if the method $m$ is only present in the safety type $T$ with type $S_1 -> \stype{T_2}{U_2}$, then the result type of the method invocation is $\stype{T_2}{\top}$ (TmH): if we bypass the declassification type, the result must be protected as a secret.

The security property obtained by this approach is called Type-based Relaxed Noninterference (TRNI). 
At the core of TRNI is a notion of observational equivalence between objects {\em up to the discrimination power of the public observer}, which is specified by the declassification type. More precisely, two objects $o_1$ and $o_2$ are equivalent at type $\stype{T}{U}$ if, for any method $m$ with 
type $S_1 -> S_2$ in the declassification type $U$, invoking $m$ with equivalent values $v_1$ and $v_2$ at type $S_1$, produces equivalent results at type $S_2$.

TRNI is formulated as a {\em modular} reasoning principle, over open terms: $\rni{\Gamma}{e}{S}$. The closing typing environment $\Gamma$ specifies the secrecy of the inputs that $e$ can use, and the security type $S$ 
specifies the observation power of the adversary on the output. 

For instance, suppose $\StringLen$ is an interface that exposes a $\mathsf{length: \Unit_\lbot -> \Int_\lbot}$ method. Then, with $\Gamma = x:\stype{\String}{\StringLen}$, the judgment $\rni{\Gamma}{\mathsf{x.length()}}{\stypebot{\Int}}$ implies:
given the knowledge that two input strings $v_1$ and $v_2$
have the same length, the lower observer does not learn anything new about the inputs by executing $\mathsf{x.length()}$. Conversely,
$\rni{\Gamma}{\mathsf{x.eq("a")}}{\stypebot{\Bool}}$ does {\em not} hold: executing $\mathsf{x.eq("a")}$ and exposing the result as a public value would reveal more information than permitted by the input declassification type ($\mathsf{eq} \notin \StringLen$). However, $\rni{\Gamma}{\mathsf{x.eq("a")}}{\stypetop{\Bool}}$ {\em does} hold, because the result is private and therefore unaccessible to the public observer.

\section{Polymorphic Relaxed Noninterference}
\label{sec:gobsec-overview}

We first motivate polymorphic declassification with faceted types, and then we illustrate the role of bounded polymorphism for
declassification. 
Finally, we give an overview of the modular reasoning principle of polymorphic relaxed noninterference.

\subsection{Polymorphic Declassification}

When informally discussing the possible extensions to their approach to declassification, 
\citet{cruzAl:ecoop2017} illustrate the potential benefits of polymorphic declassification by giving the example of a list of strings that is polymorphic with respect to the declassification type of its elements:
\begin{displaymath}
\begin{array}{rcl}
	\ListStr\langle  X \rangle & \triangleq & 
			\begin{block}
			  \lbrack~
				\mathsf{isEmpty}:\stypebot{\Unit} -> \stypebot{\Bool}, \\
				\mathsf{head}:\stypebot{\Unit} -> \stype{\String}{X}, \\
				\mathsf{tail}:\stypebot{\Unit} -> \stypebot{\goinst{\ListStr}{X}}				
				\rbrack \\
			\end{block}	
\end{array}
\end{displaymath}

This recursive polymorphic declassification policy allows a public observer to traverse the list, and to observe {\em up to} $X$ on each of the elements. This restriction is visible in the signature of the \icode{head} method, which returns a value of type $\stype{\String}{X}$.

Then, with polymorphic declassification we can implement data structures that are agnostic to the declassification policies
of their elements, as well as polymorphic methods over these data structures.
For example, we can construct declassification-polymorphic lists of strings with the following \icode{cons} method:

\begin{lstlisting}[numbers=none]
(*$\stypebot{\goinst{\ListStr}{X}}$*) cons<(*$X$*)>((*$\stype{\String}{X}$*) s, (*$\stypebot{\goinst{\ListStr}{X}}$*) l){
  return new {
	  self: (*$\stypebot{\goinst{\ListStr}{X}}$*)
	  isEmpty() => false
	  head() => s
	  tail() => l
	}
}
\end{lstlisting}
The \icode{cons} method does not even access any method of list \icode{l}, it simply returns a new declassification-polymorphic list of strings as a new object with the expected methods. We can then use this method to define a declassification-polymorphic list concatenation method \icode{concat}:
$\stypebot{\goinst{\ListStr}{X}} \times \stypebot{\goinst{\ListStr}{X}} ->
\stypebot{\goinst{\ListStr}{X}}$ defined below:

\begin{lstlisting}[numbers=none]
(*$\stypebot{\goinst{\ListStr}{X}}$*) concat<(*$X$*)>((*$\stypebot{\goinst{\ListStr}{X}}$*) l1, (*$\stypebot{\goinst{\ListStr}{X}}$*) l2){
  if(l1.isEmpty()) return l2
  return cons<(*$X$*)>(l1.head(),
                 concat<(*$X$*)>(l1.tail(),l2))
}
\end{lstlisting}
The \icode{concat} and \icode{cons} methods are standard  object-oriented implementations of list concatenation and construction, respectively. 
The \icode{concat} method respects the declassification type $\goinst{\ListStr}{X}$ of both lists because 
it uses \icode{l1.isEmpty()} and
\icode{l1.tail()} to iterate over \icode{l1}, and it uses \icode{l1.head()} 
to create a new declassification-polymorphic list of type $\stypebot{\goinst{\ListStr}{X}}$. In particular, it uses no string-specific methods. 

\subsection{Bounded Polymorphic Declassification}

The declassification interface $\goinst{\ListStr}{X}$ above is fully polymorphic, in that a public observer cannot exploit {\em a priori} any information about the elements of the list.
In particular, it is not possible to implement a polymorphic \icode{contains} method that would yield publicly-observable results. 
Indeed, \icode{contains} needs to invoke \icode{eq} over the elements of the list (obtained with \icode{head}). Because the result of \icode{head} has declassification type $X$, for {\em any} $X$, the results of equality comparisons are necessarily private.

In order to support polymorphic declassification more flexibly, we turn to {\em bounded} parametric polymorphism. Bounded parametric polymorphism supports the specification of both upper and lower bounds on type variables. 
The type $\goinst{\ListStr}{X}$ is therefore equivalent to $\goinst{\ListStr}{X : \String  ..  \top}$, where the notation $X: A..B$ is used to denote that $X$ is type variable that ranges between $A$ and $B$). Note that for $\ListStr$ to be well-formed, the declassification type variable $X$ must at least be a supertype of the safety type $\String$.

Going back to declassification-polymorphic lists, if we want to allow the definition of methods like \icode{contains}, 
we can further constrain the type variable $X$ to be a subtype of $\StringEq$: 
\begin{displaymath}
	\begin{block}
		\ListEqStr \langle X: \String .. \StringEq \rangle \triangleq \\
		\quad \quad \lbrack~
			\cdots, \mathsf{tail}:\stypebot{\Unit} -> \stypebot{\mathsf{\goinst{\ListEqStr}{\mathit{X}}}}
		\rbrack \\
	\end{block}	
\end{displaymath}
The type $\ListEqStr$ denotes a recursive polymorphic declassification policy that allows a public observer to traverse the list and compare its elements for equality with a given public element.
With this policy we can implement a generic \icode{contains} function with publicly-observable result:
\begin{lstlisting}[numbers=none]
(*$\stypebot{\Bool}$*) contains<(*$X:\String..\StringEq$*)> 
             ((*$\stypebot{\goinst{\ListEqStr}{X}}$*) l, (*$\stypebot{\String}$*) s){
  if(l.isEmpty()) return false
  if(l.head().eq(s)) return true
  return contains(l.tail(),s)
}
\end{lstlisting}
The key here is that \icode{l.head().eq(s)} is guaranteed to be publicly observable, because the actual declassification policy with which $X$ will be instantiated necessarily includes (at least) the \icode{eq} method.
Thus, upper bounds on declassification variables are useful for supporting polymorphic {\em clients}.

As mentioned above, the lower bound of a type variable used for declassification must at least be the safety type for well-formedness. 
More interestingly, the lower bound plays a critical (dual) role for {\em implementors} of declassification-polymorphic functions.
Consider a method with signature
$$\gmtype{X:\String ..\top}{\stype{\String}{\StringLen}}{\stype{\String}{X}}$$

Can this method return non-public values? For instance, can it be the identity function? No, because returning a string of type $\stype{\String}{\StringLen}$ would be unsound. Indeed, a client could instantiate $X$ with $\String$, yielding $$\stype{\String}{\StringLen} -> \stype{\String}{\String}$$
Therefore, to be sound for all possible instantiations of $X$, the implementor of the method has no choice but to return a public string. 

To recover flexibility and allow a polymorphic implementation to return non-public values, we can constrain the lower bound of $X$. For instance
$$\gmtype{X:\StringLen .. \top}{\stype{\String}{\StringLen}}{\stype{\String}{X}}$$
admits the identity function as an implementation, in addition to other implementations that produce public results. Returned values cannot be {\em more private} than specified by the lower bound of $X$; their type must be a subtype of the lower bound.

Having illustrated the interest of upper and lower bounds of declassification type variables in isolation, we now present an example that combines both.
Consider two lists of strings, each with one of the following declassification policies:
\begin{displaymath}
\begin{block}
\goinst{\ListStrLen}{X:\String .. \StringLen} \defas \\
\begin{array}{rcl}
	 &  & 
			\begin{block}
			  \lbrack~
				\mathsf{isEmpty}:\stypebot{\Unit} -> \stypebot{\Bool}, \\
				\mathsf{head}:\stypebot{\Unit} -> \stype{\String}{X}, \\
			  \mathsf{tail}:\stypebot{\Unit} -> \stypebot{\ListStrLen}
				\rbrack \\
			\end{block}	\\
	\ListStrFstLen & \defas & 
			\begin{block}
			  \lbrack~
				\mathsf{isEmpty}:\stypebot{\Unit} -> \stypebot{\Bool}, \\
				\mathsf{head}:\stypebot{\Unit} -> \stype{\String}{\StrFstLen}, \\
			  \mathsf{tail}:\stypebot{\Unit} -> \stypebot{\ListStrFstLen}
				\rbrack \\
			\end{block}	
\end{array}
\end{block}
\end{displaymath}
$\ListStrLen$ is declassification polymorphic, ensuring that at least the length of its elements is declassified ($X$ has upper bound $\StringLen$).
The second policy, $\ListStrFstLen$, is monomorphic: it declassifies both the first character and the length of its elements.
If we want a function able to concatenate these two string lists, its most general polymorphic signature ought to be:
\begin{displaymath}
\begin{block}
\langle X:\StrFstLen .. \StringLen \rangle \\
\quad\stypebot{\goinst{\ListStrLen}{X}} \times \stypebot{\ListStrFstLen}
-> \stypebot{\goinst{\ListStrLen}{X}}
\end{block}
\end{displaymath}
The upper bound $\StringLen$ is required to have a valid instantiation of $\goinst{\ListStrLen}{X}$; the lower bound $\StrFstLen$ is required to be able to add elements of the second list to the returned list.

\subsection{Reasoning principles for PRNI}

Introducing polymorphism in declassification types yields an extended notion of type-based relaxed noninterference called {\em polymorphic relaxed noninterference} (PRNI). PRNI exactly characterizes that a program with polymorphic types must be
secure for any instantiation of its type variables. To account for type variables, the judgment 
$\gtrni{\Delta}{\Gamma}{e}{S}$ is parametrized by $\DeltaX$, a
set of bounded type variables (\ie $\DeltaX ::= \cdot | \DeltaX,X:A..B $). 
As in $\rni{\Gamma}{e}{S}$, the closing typing environment $\Gamma$ specifies the secrecy of the inputs that $e$ can use, and $S$ specifies
the observation level for the output. 
$\DeltaX$ gives meaning to the type variables that can occur in both $S$ and $\Gamma$: $e$ is secure for {\em any} instantiation of type variables that respects the bounds.

For instance, given $\DeltaX \defas X:\StrFstLen..\StringLen$ and $\Gamma \defas x: \stype{\String}{X}$, 
$\gtrni{\DeltaX}{\Gamma}{\mathsf{x.length()}}{\stypebot{\Int}}$ holds because for any type $T$ such as 
$\StrFstLen <: T <: \StringLen$, and the knowledge that two input strings are related at $\stype{\String}{T}$, and hence
at $\stype{\String}{\StringLen}$ (\ie both strings have the same length), the public observer does not learn anything new
by executing \icode{x.length()}. However, $\gtrni{\DeltaX}{\Gamma}{\mathsf{x.first()}}{\stypebot{\String}}$ does not hold. Note that if
we substitute $X$ by $\StringLen$, given two strings with the same length \icode{``abc"} and \icode{``123"}, the public observer
is able to distinguish them by executing \icode{``abc".first()} and \icode{``123".first()} and observing the results
\icode{``a"} and \icode{``1"} as public values.

Also, $\gtrni{\DeltaX}{\Gamma}{\mathsf{x}}{\stype{\String}{\StringFirst}}$ does not hold. Again, we can substitute $X$ by $\StringLen$,
and take input strings \icode{``abc"} and \icode{``123"}, which can be discriminated by the public observer at type $\stype{\String}{\StringFirst}$.
However, $\gtrni{\DeltaX}{\Gamma}{\mathsf{x}}{\stype{\String}{\StringLen}}$
does hold: any two equivalent values at $\stype{\String}{T}$ where $\StrFstLen <: T<:\StringLen$ have at least the same length.

The rest of this paper dives into the formalization of polymorphic relaxed noninterference in a pure object-oriented setting (Sections~\ref{sec:gobsec-model} and~\ref{sec:gobsec-trni}), before discussing the necessary extensions to accommodate primitive types (Section~\ref{sec:gobsec-primitive-types}).

\section{Formal Semantics}
\label{sec:gobsec-model}

We model polymorphic type-based declassification in \gobsec, an extension of the 
language \obsec~\cite{cruzAl:ecoop2017} with 
polymorphic declassification. 
\obsec is based on the object calculi of \citet{abadiCardelli:1996}, 
and our treatment of type variables and bounded polymorphism is inspired by Featherweight Java~\cite{igarashi01fj} and DOT~\cite{rompfAmin:oopsla2016}. 

\subsection{Syntax}

Figure~\ref{fig:gobsec-final-syntax} presents the syntax of \gobsec. We highlight 
the extension for polymorphic declassification, compared to the syntax of \obsec. 

The language has three kind of expressions: objects, method
invocations, variables. Objects $\gobject$are collections of method
definitions. Recall that the self variable $z$ binds the current object.

A security type is a faceted type $\stype{T}{U}$, where $T$ is called the {\em safety type} of $S$, and $U$ is called the {\em declassification type} of $S$.
Types $T$ include object types $O$ and self type variables $\alpha$.
Declassification types $U$ additionally feature type variables $X$, to express polymorphic declassification. We use metavariables $A$ and $B$ for declassification type bounds.

An object type $\rtypex{\alpha}{\gntrecordx{m}{M}}$ is a collection of method signatures with 
unique names (we sometimes use $R$ to refer to just a collection of methods). The self type variable $\alpha$ binds to the defined object type (\ie object types are recursive types).
A method signature $\gmtype{X:A..B}{S_1}{S_2}$ introduces the type variable $X$ with lower bound $A$ and upper bound $B$. 
To simplify the presentation of the calculus, we model single-argument methods with a single type variable.\footnote{The implementation supports both multiple arguments and multiple type variables.}

\begin{figure}[t]
  \begin{small}		
		\begin{displaymath}
			\begin{array}{rcll}
				e & ::= & v |  \gminv{e}{m}{\Gbox{U}}{e} | x  & \text{(terms)}\\
				v & ::= & o  & \text{(values)}\\					
				o & ::= & \gobject  & \text{(objects)}\\
				S & ::= & \stype{T}{U} & \text{(security types)}\\
				T& ::= & O | \alpha &  \text{(types)}\\
				U, A, B & ::= & T | \Gbox{X} & \text{(declassification types)}\\
				O & ::= & \rtypex{\alpha}{R} & \text{(object types)}\\
				R & ::= & \gntrecordx{m}{M} & \text{(record types)}\\
				M & ::= & \gmtype{\Gbox{X:A..B}}{S}{S} & \text{(method signatures)}\\
				\Gamma & ::= & \bigcdot | \Gamma, x:S & \text{(type environments)} \\
			\GammaSub & ::= & \bigcdot | \GammaSub, \alpha <: \beta & \text{(subtyping environments)}\\
			\Gbox{\DeltaX}	 & ::= & \Gbox{\bigcdot| \DeltaX, X : A..B} & 
			\begin{block}
				\text{(type } \\
				\text{variable environments)}
			\end{block}\\
			\alpha,\beta & ~ & \text{(self type variables)} & 
			\end{array}
		\end{displaymath}
 \end{small}
 \caption{\gobsec: Syntax}
  \label{fig:gobsec-final-syntax}
\end{figure}

\subsection{Subtyping}
\label{sec:gobsec-subtyping}

Figure~\ref{fig:gobsec-subtyping} presents the \gobsec subtyping judgment $\subtp{\DeltaX; \GammaSub}{U_1}{U_2}$. 
The type variable environment $\DeltaX$ is a set of type variables with their bounds, \ie ${\DeltaX ::= \bigcdot | 
\DeltaX, X: A..B}$. The subtyping environment $\GammaSub$ is a set of subtyping assumptions between self type 
variables, \ie ${\GammaSub ::= \bigcdot | \GammaSub, \alpha <: \beta}$

The rules for the monomorphic part of the language are similar to \obsec.
Rule (SObj) justifies subtyping between two object types; it holds if the methods of 
the left object type $O_1$ are subtypes of the corresponding methods on $O_2$. Both width and depth subtyping are supported.
Note that to verify subtyping of method collections, \ie $\subtp{\DeltaX;\GammaSub, \alpha <: \beta}{R_1}{R_2}$, we 
put in subtyping relation in $\GammaSub$ the self variables $\alpha$ and $\beta$. Rule (SVar) accounts for subtyping between self type variables and it holds if such subtyping relation exists in the subtyping environment.

The rules (STrans) and  (SSubEq) justify subtyping by transitivity and type equivalence respectively. 
We consider type equivalence  up to renaming and folding/unfolding of self type variables~\cite{cruzAl:ecoop2017}.

The novel part of \gobsec are type variables, handled by rules (SGVar1) and (SGVar2). We follow the approach of \citet{rompfAmin:oopsla2016}.
Rule (SGVar1) justifies subtyping between a type variable $X$ and a type 
$B$, if $B$ is the upper bound of the type variable in $\DeltaX$. 
Rule (SGVar2) is dual to (SGVar1), justifying that $A<:X$ if $A$ is the lower bound of $X$. 

The judgment $\subtps{\DeltaX;\GammaSub}{R_1}{R_2}$ accounts for subtyping between collections of methods, and is used in rule (SObj).
The judgment $\subtps{\DeltaX;\GammaSub}{M_1}{M_2}$ denotes subtyping between method signatures.
For this judgment to hold, the type variable bounds $A'..B'$ of the supertype (on the right) must be included within the bounds $A..B$ of the subtype (on the left); this ensures that any instantiation on the right is valid on the left.
Then, in a type variable environment extended with $X:A'..B'$, standard function subtyping must hold (contravariant on the argument type, covariant on the return type).

Finally, rule (SST) accounts for subtyping between security types, which requires facets to be pointwise subtypes.

\begin{figure}[t]
\begin{small}
\framebox{$\subtp{\DeltaX; \GammaSub}{U_1}{U_2}$}
 \begin{mathpar}
    \inference[(SObj)]{
      O_1 \triangleq  \rtypex{\alpha}{R_1} \quad O_2 \triangleq \rtypex{\beta}{R_2} \\
			\subtp{\DeltaX; \GammaSub, \alpha <: \beta}{R_1}{R_2}
    }{
       \subtp{\DeltaX;\GammaSub}{O_1}{O_2}
    }\\
		\inference[(SVar)]{
      \alpha <: \beta \in \GammaSub
    }{
      \subtp{\DeltaX;\GammaSub}{\alpha}{\beta}
    } \quad
		\inference[(SSubEq)]{
      \typeeq{O_1}{O_2}
    }{
      \subtp{\Delta_X;\GammaSub}{O_1}{O_2}
    }
		\\
		\inference[(SGVar1)]{
      X : A..B \in \DeltaX
    }{
      \subtp{\DeltaX;\GammaSub}{X}{B}
    }
		\quad
		\inference[(SGVar2)]{
      X : A..B \in \DeltaX
    }{
      \subtp{\DeltaX;\GammaSub}{A}{X}
    }
		\\
		\inference[(STrans)]{
      \subtp{\DeltaX;\GammaSub}{U_1}{U_2} & \subtp{\DeltaX;\GammaSub}{U_2}{U_3}
    }{
      \subtp{\DeltaX;\GammaSub}{U_1}{U_3}
    }
\end{mathpar}
\framebox{$\subtps{\DeltaX;\GammaSub}{R_1}{R_2}$}
\begin{mathpar}
		\inference[(SR)]{
				\overline{m'} \subseteq  \overline{m} \quad m_{i} = m'_{j} \implies \subtp{\DeltaX; \GammaSub}{M}{M'}
			}{
				 \subtp{\DeltaX;\GammaSub}{\gntrecordx{m}{M}}{\gntrecordx{m'}{M'}}
			}
\end{mathpar}
\framebox{$\subtps{\DeltaX;\GammaSub}{M_1}{M_2}$}
\begin{mathpar}
		\inference[(SM)]{								
				\subtp{\DeltaX; \GammaSub}{B'}{B} \quad \subtp{\DeltaX; \GammaSub}{A}{A'} \\
				\subtp{\DeltaX,\Gbox{X:A'..B'};\GammaSub}{S'_{1}}{S_{1}} \\
				\subtp{\DeltaX,\Gbox{X:A'..B'};\GammaSub}{S_{2}}{S'_{2}}
			}{
				 \subtp{\DeltaX;\GammaSub}{\gmtype{X:A..B}{S_{1}}{S_{2}}}{\gmtype{X:A'..B'}{S'_{1}}{S'_{2}}}
			}
\end{mathpar}
\framebox{$\subtps{\DeltaX;\GammaSub}{S_1}{S_2}$}
 \begin{mathpar}
   \inference[(SST)]{
     \subtp{\DeltaX;\GammaSub}{T_1}{T_2} &    
     \subtp{\DeltaX;\GammaSub}{U_1}{U_2} &
   }{
     \subtps{\DeltaX;\GammaSub}{\stype{T_1}{U_1}}{\stype{T_2}{U_2}}
   }   
  \end{mathpar}
\end{small}
\caption{\gobsec: Subtyping rules} 
\label{fig:gobsec-subtyping}
\end{figure}

\subsection{Type System}
\label{sec:gobsec-static-semantics} 

The typing rules of \gobsec appeal to some auxiliary definitions, given in Figure~\ref{fig:gobsec-aux-definitions}.
Function $\ubound{\DeltaX}{U}$ returns the upper bound of a type $U$ in the type variable environment $\DeltaX$. 
Since \gobsec has a top type ($\rtypex{\alpha}{\rtop}$) this recursive definition of \icode{ub} is well-founded; as in Featherweight Java~\cite{igarashi01fj}, we assume that $\DeltaX$ does not contain cycles.
The auxiliary judgment $\DeltaX |- m \in U$ holds if method $m$ belongs to type $U$. For a type variable, this means that the method is in the upper bound $\ubound{\DeltaX}{X}$.
Function $\mathsf{msig}(\DeltaX,U,m)$ returns the polymorphic method signature of method $m$ in type $U$. 
The rule for type variables looks up the signature in the upper bound.
The rule for object types is standard; 
remark that it returns closed type signatures with respect to the self type variable.
Finally, the judgment $\DeltaX |- U \in A..B$ holds if the type $U$ is a super type of $A$ and a subtype of $B$ in the type variable environment $\DeltaX$.

\begin{figure}[t]%
	\begin{small}
		\framebox{$\ubound{\DeltaX}{U} = T$}
			\begin{mathpar}
				 \inference{
					T \neq X
				 }{
					 \ubound{\DeltaX}{T} = T
				 }
				\quad 
				\inference{
					X: A..B \in \DeltaX 
				}{
					\ubound{\DeltaX}{X} = \ubound{\DeltaX}{B}
				}			
			\end{mathpar}
		\framebox{$\DeltaX |- m \in U$}
		\begin{mathpar}			
			\inference{		
				O \triangleq  \rtypex{\alpha}{\gntrecordx{m}{M}} 
			 }{
				\DeltaX |- m_i \in O
			 } \quad
			 \inference{		
				\DeltaX |- m \in \ubound{\DeltaX}{X}
			 }{
				 \DeltaX |- m \in X
			 }
		\end{mathpar}
		\begin{mathpar}
		\end{mathpar}
	  \framebox{$\tlookup{\DeltaX,U}{m}{M}$}
		\begin{mathpar}			
			\quad		
			 \inference{		
				~
			 }{
				 \gtlookup{\DeltaX}{X}{m} = \gtlookup{\_}{\ubound{\DeltaX}{X}}{m}
			 } \\
			\inference{
					O \triangleq  \rtypex{\alpha}{\gntrecordx{m}{M}}
			 }{
				 \gtlookup{\_}{O}{m_i} = \ssubst{M}{O}{\alpha}
			 }
		\end{mathpar}		
		\framebox{$\DeltaX |- U \in A..B$}
		\begin{mathpar}			
			\quad		
			 \inference{		
				\subtp{\DeltaX;\bigcdot}{A}{U}  \quad \subtp{\DeltaX;\bigcdot}{U}{B}
			 }{
				 \DeltaX |- U \in A..B
			 }	
		\end{mathpar}
	\end{small}
\caption{\gobsec: Some auxiliary definitions}
\label{fig:gobsec-aux-definitions}%
\end{figure}

Figure~\ref{fig:gobsec-static-semantics} presents the typing judgment  ${\DeltaX;\Gamma \vdash e: S}$ for \gobsec, which denotes that ``expression $e$ has type $S$ under type variable environment $\DeltaX$ and 
type environment $\Gamma$''. 
Note that in our presentation, we omit well-formedness rules for types and environments. They are included in Appendix~\ref{sec:gobsec-well-formedness-rules-appendix}.

The first three typing rules are standard: rule (TVar) types a variable according to the environment, rule (TSub) is the subsumption rule
and rule (TObj) types an object. The method definitions of the object must be well-typed with respect to the method signatures taken from the safety type $T$ of the security type $S$ ascribed to the self variable $z$. 
For this, the method body $e_i$ must be well-typed in an extended type variable environment with the type variable 
$\DeltaX,X:A_{i}..B_{i}$, and an extended type environment with the self variable and the method argument.

Rules (TmD) and (TmH) cover method invocation, and account for declassification. The actual argument type $U'$
must satisfy the variable bounds $\DeltaX |- U' \in A..B$. 
On the one hand, rule (TmD) applies when the method $m$ is in $U$ with 
signature  $\gmtype{X:A..B}{S_1}{S_2}$; this corresponds to a use of the object at its declassification interface. Then, the method invocation has type $S_2$ substituting $U'$ for $X$. On the other hand, rule (TmH) applies when $m$ is not in $U$, but it is in $T$; this corresponds to a use beyond declassification and should raise the security to high. This is why the result type is $\stype{\ssubst{T_2}{X}{U'}}{\top}$. This is all similar to the non-polymorphic rules (Figure~\ref{fig:obsec-syntax}), save for the type bounds check, and the type-level substitution.

\begin{figure}[t]
\begin{small}
\framebox{$\DeltaX; \Gamma |- e : S$}
\begin{mathpar}	
	\inference[(TVar)]{
				x \in dom(\Gamma)
			}{
				\DeltaX; \Gamma |- x: \Gamma(x)
			}\
  \inference[(TSub)]{
			\DeltaX;\Gamma |- e: S' \quad \subtps{\DeltaX;\bigcdot}{S'}{S}
    }{
      \DeltaX;\Gamma |- e: S
		}\\
  \inference[(TObj)]{
			S \triangleq \stype{T}{U} \quad 
			\tlookup{\_,T}{m_i}{\gmtype{X:A_i..B_i}{S'_i}{S''_i}} \\
			\DeltaX,X:A_i..B_i;\Gamma, z: S, x:S^{'}_i |- e_i : {S''_i}
    }{
      \DeltaX;\Gamma |- \objectx{z}{S}: S
    }
\end{mathpar}

\begin{mathpar}
  \inference[(TmD)]{
      \DeltaX;\Gamma |- e_1 : \stype{T}{U} & 
			\DeltaX |- m \in U \\ 
			\tlookup{\DeltaX;U}{m}{\gmtype{X:A..B}{S_1}{S_2}}	\\
			\DeltaX |- U' \in A..B
			&
			\DeltaX;\Gamma |- e_2 :  \ssubst{S_1}{U'}{X}
    }{
      \DeltaX;\Gamma |- \gminv{e_1}{m}{U'}{e_2}: \ssubst{S_2}{U'}{X}
    }		
\end{mathpar}
\begin{mathpar}
  \inference[(TmH)]{
      \DeltaX;\Gamma |- e_1 : \stype{T}{U} &
      \DeltaX |- m  \notin U \\
			\tlookup{\DeltaX;T}{m}{\gmtype{X:A..B}{S_1}{\stype{T_2}{U_2}}} \\
			\DeltaX |- U' \in A..B
			&
			\DeltaX;\Gamma |- e_2 : \ssubst{S_1}{U'}{X}
    }{
      \DeltaX;\Gamma |- \gminv{e_1}{m}{U'}{e_2}: \stype{\ssubst{T_2}{U'}{X}}{\top}			
			}
\end{mathpar}
\end{small}
\caption{\gobsec: Static semantics}
  \label{fig:gobsec-static-semantics}
\end{figure}

\subsection{Dynamic Semantics}
\label{sec:gobsec-dynamic-semantics}

The small-step dynamic semantics of \gobsec are standard, given in Figure~\ref{fig:gobsec-dynamic-semantics}. They rely 
on evaluation contexts and use the auxiliary function $\mathsf{methimpl}(o.m)$ to lookup a method implementation.
Note that types in general, and type variables in particular, do not play any role at runtime.

\begin{figure}[t]
\begin{small}
	\framebox{$\methimpl{o}{m}{x.e}$}
			\begin{mathpar}
				 \inference{		
					o \triangleq \objectx{z}{S}
				 }{
					 \methimpl{o}{m_i}{x.e_i}
				 }	
			\end{mathpar}
	\begin{mathpar}			
			\begin{array}{llll}
				E & ::= & \left[~\right] | E.m(e) |  v.m(E)& \text{(evaluation contexts)}
			\end{array}\\
      \inference[(EMInvO)]{
        o \triangleq \objectxx{z}{\_}{\_} \quad \methimpl{o}{m}{x.e}
      }{
        E[\gminv{o}{m}{\_}{v}] \reduce E[\mbsubst{e}{o}{z}{v}{x}]
      }			
	\end{mathpar} 
\end{small}
\caption{\gobsec: Dynamic semantics}
\label{fig:gobsec-dynamic-semantics}
\end{figure}

\subsection{Safety}
\label{sec:gobsec-safety}
We first define what it means for a closed expression $e$ to be \emph{safe}: an expression is safe if it evaluates to a value, or diverges without getting stuck.
\begin{restatable}[Safety]{definition}{gobsecSafe}
\label{def:gobsec-safe}
$\mathsf{safe}(e) \Longleftrightarrow \forall e'.~ e \reduce^{*} e' \implies e' = v~ or ~\exists e''.~ e' \reduce e''$
\end{restatable}
Well-typed \gobsec closed terms are safe.
\begin{restatable}[Syntactic type safety]{theorem}{gobsecTypeSafety}
\label{the:gobsec-type-safety}
$ |- e : S  \implies  \mathsf{safe}(e)$
\end{restatable}
But of course, type safety is far from sufficient; we want to make sure that
well-typed \gobsec terms are {\em secure}. To this end, the next section formalizes the precise notion of security we consider in \gobsec, and proves that it is implied by typing.

\section{Polymorphic Relaxed Noninterference, Formally}
\label{sec:gobsec-trni}

We now formally define the security property of \emph{polymorphic type-based relaxed noninterference} (PRNI), and prove that
 the \gobsec type system soundly
enforces PRNI.

\subsection{Logical Relation}
\label{sec:gobsec-security-definitions}

We define how values, terms and environments are related 
through a step-indexed logical relation~\cite{ahmed:esop2006} 
(Figure~\ref{fig:gobsec-logical-relation}). 
Step-indexing is needed to ensure that the logical relation is
well-founded in presence of recursive object types.

The main novelty of this logical relation with respect that of \obsec is 
that it needs to give an interpretation to polymorphic security types of the form $\stype{T}{X}$. We do this by quantifying over all possible actual types $U$ for $X$ and interpreting $\stype{T}{U}$. The interpretation of a security type is expressed as sets of \emph{atoms} of the form $(k,e_1,e_2)$, where $k$ is a {\em step index} meaning that $e_1$ and $e_2$ are related
for $k$ steps. 

The definition also appeals to a {\em simple} typing judgment ${\stypeof{\Gamma}{e}{T}}$, which disregards the declassification types, and is therefore standard (Appendix~\ref{sec:gobsec-simple-ts-appendix}). 
We have that 
$\DeltaX;\Gamma |- e: \stype{T}{U} => \stypeof{\Gamma}{e}{T}$.
The use of this simple type system in the logical relation 
clearly separates the definitions of security from its
static enforcement by the type system
of \S\ref{sec:gobsec-static-semantics}~\cite{cruzAl:ecoop2017}.

The logical relation uses several auxiliary definitions. 
$\atomone{n}{T}$ requires $e_1$ and $e_2$ to be \emph{simply well-typed} expressions of type $T$ and the index $k$ to be strictly less than $n$. 
$\atomvalone{n}{T}$ restricts $\atomone{n}{T}$ to values.
$\atomunion{T}$ are atoms of simply well-typed expressions of type $T$ 
(\ie for {\em any} step-index $k$).

The definition of $\nsetv{\stype{T}{O}}$ relates two objects $o_1, o_2$ for $k$ steps if for 
any method $m \in O$ and with signature $\gmtype{X:A..B}{S'}{S''}$ and $j<k$, given related arguments for $j$ steps at $S'$,
invocations of $m$ produce related results for $j$ steps at $S''$.
More specifically, given {\em any} actual 
type $T'$ that satisfies the bounds of the type parameter $X$ (i.e., $T' \in A..B$)
and given related arguments in 
$\nsetv{\ssubst{S'}{T'}{X}}$ we must obtain related computations in $\nsetv{\ssubst{S''}{T'}{X}}$.

The relational interpretation of expressions $\nsetc{\stype{T}{U}}$ relates atoms of the form $(k,e_1,e_2)$ that satisfy
that for all $j < k$, if both  expressions $e_1$ and $e_2$ reduce to values $v_1$ and $v_2$ in at most $k$ steps 
then $v_1$ and $v_2$ must be equivalent for the remaining $k-j$ steps. 
This definition is termination-insensitive: if 
one expression does not terminate in less that $k$ steps,
then the two expressions are trivially equivalent.

Type environments have standard interpretations. 
$\nsetg{\Gamma}$ relates value substitutions $\gamma$, \ie
mappings from variables to closed values, as triples of the form $(k,\gamma_1,\gamma_2)$, where 
$\gamma_1$ and $\gamma_2$ are related if they have the same variables as $\Gamma$, and for any variable $x$, the associated values 
are related for $k$ steps at type $\Gamma(x)$.
Finally, a type substitution $\sigma$, \ie a mapping from type variables to closed types, {\em satisfies} a type variable environment $\DeltaX$, noted $\nsetd{\DeltaX}$, if it 
has the same type variables that $\DeltaX$ and the mapped type $T$ is within the type variable bounds.

\begin{figure*}[t]
	\begin{small}
	\begin{displaymath}
		\begin{array}{lcl}
			\atomone{n}{T} & = & \{(k,e_1,e_2) | k < n ~\wedge ~|-_1 e_1: T~\wedge~|-_1 e_2:T\} \\
			\atomvalone{n}{T} & = & \{(k,v_1,v_2) \in \atomone{n}{T} \} \\
			\atomunion{T} & = & \{ (k,e_1,e_2) \in \bigcup\limits_{n \ge 0} \atomone{n}{T} \}\\
			\nsetv{\stype{T}{O}}  &  = &
			\begin{block}
				\{ (k, v_1, v_2) \in \atomunion{T} | 
				\\
				(\forall m \in O . \quad
					\tlookup{O}{m}{\gmtype{X:A..B}{S'}{S''}}					
					\\
					\quad \forall j < k,\Gbox{T'}, v_1',v_2'. \quad \Gbox{|- T' ~\wedge~ T' \in A..B} ~\wedge~ \\
					\quad (j,v_1,v_2) \in \nsetv{\stype{T}{O}} ~\wedge~ (j, v_1', v_2') \in \nsetv{\ssubst{S'}{T'}{X}} \implies
					\\					 
					\quad \quad (j, \gminv{v_1}{m}{\_}{v'_1}, \gminv{v_2}{m}{\_}{v'_2}) \in \nsetc{\ssubst{S''}{T'}{X}} )
				\}
			\end{block} \\
			\nsetc{\stype{T}{U}} &  =  & 
			\begin{block}			 
				\{ (k, e_1, e_2) \in \atomunion{T}|~
					(\forall j < k.
					 (e_1 \reduce^{\le j} v_1~\wedge~e_2 \reduce^{\le j} v_2 )
					\implies (k-j,v_1,v_2) \in \nsetv{\stype{T}{U}})
				\}
			\end{block} \\
			\nsetg{\cdot} & = & \{(k,\emptyset,\emptyset)\} \\
			\nsetg{\Gamma;x:S} & = &
				\begin{block}					
					\{(k,\gamma_1\left[x \mapsto v_1 \right],\gamma_2\left[x \mapsto v_2 \right])|~
					(k,\gamma_1,\gamma_2) \in \nsetg{\Gamma}
					~\wedge
					~ (k,v_1, v_2) \in \nsetv{S}
					\} 
				\end{block} \\
			\nsetd{\cdot} & = & \{\emptyset\} \\
			\nsetd{\DeltaX;X:A..B} & = & 
			\begin{block}
				\{\sigma\left[X \mapsto T \right] |
					~\sigma \in \nsetd{\DeltaX} \wedge \DeltaX |- T \in A..B\} \\
			\end{block}
    \end{array}		
  \end{displaymath}	
	\end{small}
 \caption{\gobsec Step-indexed logical relation for type-based equivalence}
  \label{fig:gobsec-logical-relation}
\end{figure*}

\subsection{Defining Polymorphic Relaxed Noninterference}
Having defined the logical relation, we can now formally define PRNI. 
As standard, noninterference properties allow modular reasoning about open terms with respect to (term-level) variables. For PRNI, we extend this modular reasoning principle to open terms with respect to {\em type} variables.
Then, a simply well-typed expression $e$ under $\DeltaX$ and a well-formed $\Gamma$ satisfies PRNI at well-formed type $S$,
written $\gtrni{\DeltaX}{\Gamma}{e}{S}$, if for any type substitution $\sigma$ that satisfies $\DeltaX$ and 
two value substitutions $\gamma_1$ and $\gamma_2$ in the relational interpretation of $\sigma(\Gamma)$, applying the two 
value substitutions to the expression $e$ produces equivalent expressions at type $\sigma(S)$.  
As usual, the definition quantifies universally on the step index $k$.
We need only consider a single type substitution $\sigma$; indeed, type variables happen only in declassification types, which express the observation power of the public observer. Therefore, for each security type of the form $\stype{T}{X}$ we only need to consider {\em one} actual type $U$ within the bounds of $X$ to pick the observation power of the public observer. The substitution $\sigma$ captures all these choices.

\begin{definition}[Polymorphic relaxed noninterference]
\label{def:gobsec-prni}
\begin{displaymath}			
		\begin{array}{l}
			\gtrni{\DeltaX}{\Gamma}{e}{S} \Longleftrightarrow \\ 
			\quad S \triangleq \stype{T}{U} \quad \stypeof{\Gamma}{e}{T}~\wedge~ \wfex{\DeltaX}{\Gamma}~\wedge~ \DeltaX |- S~\wedge \\
			\quad \forall k \ge 0 .~\forall \sigma, \gamma_1,\gamma_2.
			~\sigma \in \nsetd{\DeltaX}. 
			~ (k,\gamma_1,\gamma_2) \in \nsetg{\sigma(\Gamma)} 
			\\
			\quad \implies (k,\sigma(\gamma_1(e)),\sigma(\gamma_2(e))) \in \nsetc{\sigma(S)}
		\end{array}
\end{displaymath}
\end{definition}

This definition captures the intuitive security notion that an expression is secure if it produces indistinguishable outputs 
up to the declassification power of the public observer (specified by $S$), when linked with indistinguishable inputs up to their declassification (specified by $\Gamma$).

\subsection{Security Type Soundness}

To establish that all well-typed \gobsec terms satisfy PRNI, we first introduce a notion of logically-related open terms, and prove 
that if an expression is related to itself, then it satisfies PRNI. We then prove the fundamental property of the logical relation, which 
states that well-typed terms are logically-related to themselves.

Two open expressions $e_1$ and $e_2$ are logically related at type $S$ in environments $\DeltaX$ and $\Gamma$ if, given a 
type substitution $\sigma$ satisfying $\DeltaX$ and value substitutions $\gamma_1$ and $\gamma_2$ in the relational interpretation of $\sigma(\Gamma)$, closing 
these expressions with the given substitutions produces related expressions related at type $\sigma(S)$.

\begin{definition}[Logical relatedness of open terms]
\label{def:gobsec-expression-equivalence}
\begin{displaymath}
	\begin{array}{l}
		\DeltaX ;\Gamma |- e_1 \approx e_2 : S   \Longleftrightarrow \\
		\quad \DeltaX;\Gamma |- e_i: S~\wedge~\wfex{\DeltaX}{\Gamma}~\wedge~ \DeltaX |- S~\wedge \\
		\quad \forall k \ge 0 .~\forall \sigma, \gamma_1,\gamma_2.
		~\sigma \in \nsetd{\DeltaX}. \\
		\quad \quad (k,\gamma_1, \gamma_2) \in \nsetg{\sigma(\Gamma)} \\
		\quad \quad \implies (k, \sigma(\gamma_1(e_1)), \sigma(\gamma_2(e_2)) \in \nsetc{\sigma(S)}
	\end{array}
\end{displaymath}
\end{definition}

Trivially, if an expression is logically related to itself, then it satisfies PRNI.

\begin{restatable}[Self logical relation implies PRNI]{lemma}{fpImpliesPRNI}
\label{def:logapprox-trni}
\mbox{}
\\
$\DeltaX;\Gamma |- e \approx e : S \implies \gtrni{\DeltaX}{\Gamma}{e}{S}$
\end{restatable}

We then turn to proving that all well-typed terms are logically-related to themselves, \ie the fundamental property of the logical relation. 

\begin{restatable}[Fundamental property]{theorem}{gfprni}
\label{lm:gobsec-fp}
\mbox{}
\\
$\DeltaX;\Gamma \vdash e : S \implies \DeltaX;\Gamma |- e \approx e : S$
\end{restatable}
\begin{proof}
The proof is by induction on the typing derivation of $e$. We use the common approach of defining compatibility lemmas for each typing rule~\cite{ahmed:esop2006}. Each case follows from the corresponding compatibility lemma. 
\end{proof}

Security type soundness follows directly from Theorem~\ref{lm:gobsec-fp} and Lemma~\ref{def:logapprox-trni}.

\begin{restatable}[Security type soundness]{theorem}{gsound}
\label{the:gsound}
\mbox{}
\\
$\DeltaX;\Gamma \vdash e : S \implies \gtrni{\DeltaX}{\Gamma}{e}{S}$
\end{restatable}%

Having proven that well-typed \gobsec programs are secure, we are almost ready to revisit the examples of Section~\ref{sec:gobsec-overview} to illustrate PRNI. We must first address the case of primitive types, discussed next.

\section{Ad-hoc Polymorphism for Primitive Types}
\label{sec:gobsec-primitive-types}

Both \obsec~\cite{cruzAl:ecoop2017} and \gobsec (so far) ignore the case of primitive types, such as integers and strings. However, in an object-oriented language, 
primitive types are both necessary and challenging from a security viewpoint. In particular, integrating type-based declassification with faceted types requires appealing to a form of {\em ad hoc} polymorphism.

\subsection{The Need and Challenge of Primitive Types}
\label{sec:gobsec-pt-overview}

In a pure object-oriented calculus (as in a pure functional calculus) without primitive types, the only real observation that can be made on programs is termination. A termination-insensitive notion of noninterference is therefore useless in a pure setting: one needs some primitive types with a purely syntactic notion of equality. Indeed, all the examples we presented in earlier sections assume a syntactic notion of observation for strings and integers.

Introducing primitive types calls for some form of label polymorphism. Indeed, we do not want to fix the security level of primitive operations, as this would be either impractical for the public observer (if all security labels were high) or for the secret observer (if all security labels were low).
This is why practical security-typed languages like FlowCaml~\cite{pottierSimonet:toplas2003} and Jif~\cite{myers:jif} use 
label-polymorphic primitive operators, specifying that the return label is the {\em least upper bound} of the argument labels. For instance, 
a binary integer operator would have type $\forall \ell_1, \ell_2. \Int_{\ell_1} \times \Int_{\ell_2} -> \Int_{\ell_1 \join \ell_2}$.
In a monomorphic security language, the same principle is hardcoded in the typing rules for primitive operators~\cite{zdancewic}.

Unfortunately this approach does not work with labels-as-types, even in a label-monomorphic setting. Indeed,
because labels are types, returning the join of the argument security labels means computing the {\em subtyping join} (denoted $\join_{<:}$ hereafter) of the declassification types. 
This is both impractical, incorrect, and potentially unsound from a security viewpoint:
\begin{itemize}
	\item {\em Impractical.} Consider a function of type $\stype{\Bool}{X_1} \times \stype{\Int}{X_2} -> \stype{\Bool}{(X_1 \join_{<:} X_2)}$. Given two public arguments (\ie~$X_1=
\Bool$, $X_2=\Int$), then assuming $\Bool \join_{<:} \Int = \top$, 
the result is necessarily private. While sound, this is way too conservative; it is impractical for primitive operations to always return private values even when given public inputs.
	\item {\em Incorrect.} Consider an integer comparison of type $\stype{\Int}{X_1} \times \stype{\Int}{X_2} -> \stype{\Bool}{(X_1 \join_{<:} X_2)}$. 
If we instantiate this signature with $\Int$ and $\Int$ we obtain an {\em ill-formed} return type, $\stype{\Bool}{\Int}$.
	\item {\em Insecure.} Consider a unary integer operator $\stype{\Int}{X} -> \stype{\Int}{X}$; this signature is not sound security-wise for all unary integer operators. Take an operator that trims the most-significant bit of its argument. If one declassifies only the parity of the argument, two equivalent inputs will not always yield two equivalent outputs (as the parity of the trimmed values might differ).
\end{itemize}

\subsection{Sound Signatures for Primitive Types}
\label{sec:gobsec-primitive-signatures}

The observations above reveal one of the flip sides of expressive declassification policies: because declassification is captured semantically, declassification polymorphism is a very strong notion compared to standard label polymorphism. In the general case, without appealing to intricate semantic conditions, 
there are therefore only two simple syntactic principles to define sound signatures for primitive operators:\footnote{The syntactic principles (P1) and (P2) are formally justified 
by the proof of Lemma~\ref{lm:syntactic-implies-semantic}, discussed in Section~\ref{fig:gobsec-pt-dynamic-semantics}.}

\begin{enumerate}
	\item[(P1)] if every argument is public (\eg $\String_{\lbot}$) then the return type can be public.
	\item[(P2)] if any argument is not public (\eg $\stype{\String}{\StringFirst}$) then the return type must be secret (\eg~$\String_{\ltop}$).
\end{enumerate}

As typical, we provide an object-oriented interface for primitive types (\eg$a+b$ is $a.+(b)$ as in Scala for instance). Therefore the principles above must be extended to account for the status of the {\em receiver} object:
if the primitive method invoked on the primitive value is part of its declassification type, then it is considered a public ``argument''; otherwise, it is private.

Note that, without any form of polymorphism, \ie~picking a single syntactic principle above, primitive types would be impractical. Duplicating all definitions to offer both options is also not a viable approach. 

\subsection{Polymorphic Primitive Signatures}

To support the two syntactic principles exposed above, we propose to use \emph{ad-hoc} polymorphism (akin to overloading) for primitive types $P$ in \gobsec.
We introduce \emph{polymorphic primitive signatures}, written $\stype{\primt}{\ilab} -> \stype{\primt}{\ilab}$. A primitive security type $\stype{\primt}{\ilab}$ is resolved polymorphically {\em at use site}, following principles (P1) and (P2) above.
Object-oriented interfaces for primitive types are exclusively composed of polymorphic primitive signatures. For instance, in \gobsec strings are primitives, declared by the following $\String$ type: 
\begin{displaymath}
\begin{array}{rcl}
		\String & \defas & 
		\begin{block}
			  \lbrack~
				\mathsf{concat}:\stype{\String}{\ilab} -> \stype{\String}{\ilab},\\
				\mathsf{first}: \stype{\Unit}{\ilab} \rightarrow \stype{\String}{\ilab}, \\
				\mathsf{length}: \stype{\Unit}{\ilab} \rightarrow \stype{\Int}{\ilab}, \\
				\mathsf{eq}: \stype{\String}{\ilab}  \rightarrow \stype{\Bool}{\ilab}, \\
				\cdots
				\rbrack
		\end{block}
\end{array}
\end{displaymath}

To illustrate, assume $\mathsf{a}:\stypebot{\String}$, $\mathsf{b}:\stypebot{\String}$ and $\mathsf{c} :\stypetop{\String}$. 
Then $\mathsf{a.eq(b)}$ has type $\stypebot{\Bool}$, while $\mathsf{a.eq(c)}$ has type $\stypetop{\Bool}$.

Primitive types can also be subject to declassification policies. For instance, consider:
\[\StringEqPoly \defas \lbrack \mathsf{eq}: \stype{\String}{\ilab}  \rightarrow \stype{\Bool}{\ilab} \rbrack
\]

\noindent and $\mathsf{d}:\stype{\String}{\StringEqPoly}$. Then $\mathsf{d.eq(b)}$ has type $\stypebot{\Bool}$, while $\mathsf{d.concat(a)}:\stypetop{\String}$.

Furthermore, one can use any type signature in a declassification policy for a primitive type, as long as it is sound. For instance, 
$\StringEqL \defas \lbrack \mathsf{eq} : \stypebot{\String} -> \stypebot{\Bool} \rbrack$ respects principle 1). Conversely, $\StringEqBad \defas \lbrack \mathsf{eq} : \stypetop{\String} -> \stypebot{\Bool} \rbrack$ cannot be used as it would violate the soundness principles above (in \gobsec, $\stype{\String}{\StringEqBad}$ is ill-formed).

\subsection{Formal Semantics}
We now formalize the treatment of primitive types in \gobsec. Figure~\ref{fig:gobsec-extend-pt} presents the extended syntax to support primitive values $\primb$ 
and primitive types $\primt$. Expression $e.m(e)$ is for method invocation on primitives; as explained previously, primitive types expose an object-oriented interface, so $a + b$ is $a.+(b)$. We extend the category $S$ with security types of the form $\stype{\primt}{\ilab}$ and 
introduce a new category $I$, for primitive signatures $\stype{\primt}{\ilab} -> \stype{\primt}{\ilab}$.  
The security type $\stype{\primt}{\ilab}$ can be used for standard signatures $\gmtype{X:A..B}{S_1}{S_2}$ as well.

\begin{figure}[t]
  \begin{small}		
		\begin{displaymath}
			\begin{array}{rcll}
				e & ::= & \cdots | e.m(e) & \text{(terms)}\\
				v & ::= & \cdots | \primb & \text{(values)}\\
				T & ::= & \cdots | \primt &  \text{(types)}\\				
				M & ::= & \cdots | I & \text{(method signatures)}\\
				S & ::= & \cdots | \stype{\primt}{\ilab} &  \text{(security types)}\\								
				I & ::= & \stype{\primt}{\ilab} -> \stype{\primt}{\ilab} &  \text{(primitive signatures)}\\								
				\primt & ::= & \text{(eg. Int, String)} & \text{(primitive types)}\\
				\GammaSub & ::= & \cdots | \GammaSub, \primt <: \beta & \text{(subtyping environments)}\\
			\end{array}
		\end{displaymath}
 \end{small}
 \caption{\gobsec: Extended syntax for primitive types}
  \label{fig:gobsec-extend-pt}
\end{figure}

\begin{figure}[t]
\begin{small}
\framebox{$\subtp{\DeltaX; \GammaSub}{U_1}{U_2}$}
 \begin{mathpar}
		
		\cdots \and \inference[(SPObj)]{
      \methods(\primt) = R_1 \quad O \triangleq \rtypex{\beta}{R_2} \\			
			\subtp{\DeltaX;\GammaSub, \primt <: \beta}{R_1}{R_2} \\
    }{
       \subtp{\DeltaX;\GammaSub}{\primt}{O}
    }\\
    \inference[(SPVar)]{
      \primt <: \beta \in \GammaSub
    }{
      \subtp{\DeltaX;\GammaSub}{\primt}{\beta}
    }
\end{mathpar}
\framebox{$\subtp{\DeltaX; \GammaSub}{M_1}{M_2}$}
 \begin{mathpar}
		\cdots \quad
		\inference[(SImpl)]{
    }{
       \subtp{\DeltaX;\GammaSub}{I}{I}
    }
\end{mathpar}

\end{small}
\caption{\gobsec: Subtyping rules for primitive types} 
\label{fig:gobsec-pt-subtyping}
\end{figure}

The changes to subtyping are in Figure~\ref{fig:gobsec-pt-subtyping}. Now, subtyping assumptions can be also made 
between a primitive type and a self type variable, \ie ${\GammaSub ::= \bigcdot | \GammaSub, \alpha <: \beta | \GammaSub, \primt <: \beta}$, and function  $\methods$ returns the methods of a primitive type.
Rule (SPObj) justifies subtyping between a primitive type and an object type, and it is very similar to rule (SObj) of 
Figure~\ref{fig:gobsec-subtyping} for object types. Rule (SPVar) accounts for subtyping between primitive types and type variables 
and it holds if such subtyping relation exists in the subtyping environment $\GammaSub$. Note that there is no rule for subtyping between an object type and a primitive type, because this would not be sound. 
Rule (SImpl) accounts for subtyping between the same primitive signature. There is no rule to justify subtyping between a primitive 
signature and a standard signature.

As we discussed at the end of Section~\ref{sec:gobsec-primitive-signatures}, we need an extra condition for the well-formedness of security types to ensure sound declassification. Given the type $\stype{T}{U}$, if $T$ 
has a method $m:I$, the method signature of $m$ in $U$ must be either  the same primitive signature $I$,  
or a normal signature that is \emph{sound}. We use the predicate $\mathsf{soundsig}$ to express that signature
$\gmtype{X:A..B}{S_1}{S_2}$ is sound, which must satisfy that either the argument type is public, or the return type is private:
\begin{displaymath}
		\begin{array}{lcl}
			\mathsf{soundsig}(\gmtype{\_}{\stype{\primt}{U_1}}{\stype{T_2}{U_2}}) & <=> & U_1=\primt \vee  U_2 = \top
		\end{array}
\end{displaymath}

Figure ~\ref{fig:gobsec-pt-static-semantics} presents the extension to the typing rules of \gobsec. Rule (TPrim) justifies typing for primitive values, using a function $\Theta$ that specifies each primitive type.
The new typing rules (TPmD) and (TPmH) realize ad hoc polymorphism for primitive types.
Rule (TPmD) is key: it applies when $m$ is in the declassification type $U$, and uses the function $\rpolicy$ to calculate the declassification type of the return type, based on the type of the argument: if the argument is public, so is the returned value.
Rule (TPmH) applies when $m$ is not in the declassification type $U$, and similarly to (TmH), ensures that the returned value is private.

\begin{figure}[t]
\begin{small}
	\framebox{$\DeltaX; \Gamma |- e : S$}
	\begin{mathpar}
		\cdots \quad 
		\inference[(TPrim)]{
			 \primt = \Theta(\primb)
			}{
			 \DeltaX;\Gamma |- \primb: \stype{\primt}{\primt}
		}\\
		 \inference[(TPmD)]{
				\DeltaX;\Gamma |- e_1 : \stype{T}{U} & \DeltaX |- m \in U \\
				\tlookup{\DeltaX,U}{m}{\stype{\primt_1}{\ilab} -> \stype{\primt_2}{\ilab}}	\\ 
				\DeltaX;\Gamma |- e_2 : \stype{\primt_1}{U_1} \\
				\rpolicy(\stype{\primt_1}{U_1},\primt_2) = \primt'_2
			}{
				\DeltaX;\Gamma |- \minv{e_1}{m}{e_2}: \stype{\primt_2}{\primt'_2}
			}\\
			\inference[(TPmH)]{
				\DeltaX;\Gamma |- e_1 : \stype{T}{U}  & \DeltaX |- m \notin U\\
				\tlookup{\DeltaX,T}{m}{\stype{\primt_1}{\ilab} -> \stype{\primt_2}{\ilab}}	\\				
				\DeltaX;\Gamma |- e_2 : \stype{\primt_1}{U_1}
			}{
				\DeltaX;\Gamma |- \minv{e_1}{m}{e_2}: \stype{\primt_2}{\top}
			}
	\end{mathpar}
	\framebox{$\rpolicy(\stype{\primt_1}{U_1},\primt_1) = U$}
	\begin{mathpar}
		\begin{array}{rl}
			\rpolicy(\stype{\primt_1}{U_1},\primt_2) = & 
			\begin{cases}
							\primt_2 & \primt_1 = U_1\\
							\top &	\textit{otherwise} \\
			\end{cases}			
		\end{array}
	\end{mathpar}
\end{small}
\caption{\gobsec: Extended static semantics for primitive types}
  \label{fig:gobsec-pt-static-semantics}
\end{figure}

Figure~\ref{fig:gobsec-pt-dynamic-semantics} shows the extension to the dynamic semantics to support primitive values. Rule (EMInvP) executes a method 
invocation on a primitive value using the 
function ${\theta}$, which abstracts over the internal implementation of primitive values. 

To prove type safety for \gobsec with primitives, we only need to assume that 
$\Theta$ and $\theta$---which are parameters of the language---agree on the specification and implementation of all primitive types and their operations.

\begin{figure}[t]
\begin{small}
	\begin{mathpar}
	 
	  \inference[(EMInvP)]{
		  ~
		}{
			E[\primb_1.m(\primb_2)] \reduce E[\theta(m,\primb_1,\primb_2)]
		}
	\end{mathpar} 
\end{small}
\caption{\gobsec: Dynamic semantics of primitive values}
\label{fig:gobsec-pt-dynamic-semantics}
\end{figure}

\subsection{Logical Relation for Primitive Types}

Figure~\ref{fig:gobsec-pt-logical-relation} presents the extension to the logical relation of Figure~\ref{fig:gobsec-logical-relation} to account for primitive types. First, $\nsetv{\stype{\primt}{\primt}}$ relates {\em syntactically equal} primitive values of type $\primt$. Second, 
the definition of $\nsetv{\stype{T}{O}}$ now accounts for primitive values that are observed with a declassification type $O$.
$\nsetv{\stype{T}{O}}$ still relates values $v_1$ and $v_2$ if, for all methods of $O$, given related arguments, the invocations of $m$ on 
$v_1$ and $v_2$ produce related computations. However, the definition now discriminates between each type of signatures.
For a method $m$ with primitive signature $\stype{\primt_1}{\ilab} -> \stype{\primt_2}{\ilab}$, we require one of the following
conditions to hold.
If we get related arguments $v'_1$ and $v'_2$ at 
$\stype{\primt_1}{\primt_1}$ (\ie public values), method invocations $v_1.m(v'_1)$ and $v_2.m(v'_2)$ need to be related at 
the public type $\stype{\primt_2}{\primt_2}$. 
Otherwise, if the arguments $v'_1$ and $v'_2$ are related at a non-public type 
($\stype{\primt_1}{U},~\primt_1 \neq U$), then $v_1.m(v'_1)$ and $v_2.m(v'_2)$ need to be related at the top type $\stype{\primt}{\top}$.
These conditions are expressed in the definition by requiring related method invocations in $\nsetc{\stype{\primt_2}{\rpolicy(\stype{\primt_1}{U_1},\primt_2)}}$.

\begin{figure}[t]%
	\begin{small}
	\begin{minipage}{0.50\textwidth}
		\begin{displaymath}
    \begin{array}{ll}			
			\nsetv{\stype{\primt}{\primt}} = & \{(k,\primb,\primb) \in \atomunion{\primt}\} \\
			\nsetv{\stype{T}{O}} = & \		  
			\begin{block}
				\{ (k, v_1, v_2) \in \atomunion{T} | \ldots
				\\ \forall m \in O . \quad \tlookup{\bigcdot,O}{m}{\stype{\primt_1}{\ilab} -> \stype{\primt_2}{\ilab}}
				\\ \forall j < k,v_1',v_2'.  U_1 >: \primt_1\\
				((j,v'_1,v'_2) \in \nsetv{\stype{\primt_1}{U_1}} \implies \\
				(j,v_1.m(v'_1),v_2.m(v'_2)) \in \nsetc{\stype{\primt_2}{\rpolicy(\stype{\primt_1}{U_1},\primt_2)}})
				\}
			\end{block} \\
		\end{array}
	\end{displaymath}
	\end{minipage}
	\end{small}
\caption{\gobsec: Step-indexed logical relation with new definitions for primitive types}%
\label{fig:gobsec-pt-logical-relation}%
\end{figure}

Extending the fundamental property (Lemma~\ref{lm:gobsec-fp}) for primitive types requires the following lemma, which states that syntactically-equal primitive values of type $\primt$ are in the object-oriented interpretation of any type $\stype{\primt}{O}$---essentially, equal values cannot be discriminated.

\begin{restatable}[Equal values are logically related]{lemma}{syntacticImpliesSemantic}
\label{lm:syntactic-implies-semantic}
\mbox{}\\
$\forall k \geq 0, \primb, \primt, O.\\
~~~~~~~~~~~~~\stypeof{}{\primb}{O}~\wedge~\primt<:O \implies  (k,\primb,\primb) \in \nsetv{\stype{\primt}{O}}$
\end{restatable}
\begin{proof}
Because $\stype{\primt}{O}$ is a well-formed security type, $O$ consists of primitive signatures and standard, sound signatures. 
For the case of primitive signatures, at some point we have $\stype{\primt_1}{*}-> \stype{\primt_2}{*}$ and
equivalent values at $(j,v_1,v_2) \in \stype{\primt_1}{U_1}$ and we have to prove that
$(j, \delta(m,b,v_1),\delta(m,b,v_2) \in \nsetc{\stype{\primt_2}{\rpolicy(\stype{\primt_1}{U_1},\primt_2)}}$.
We do case analysis on $U$. If $U_1 = \primt_1$, we know that $v_1 = v_2$ and hence if $\delta(m,b,v_1)$ and
$\delta(m,b,v_2)$ are defined, their results are syntactically equal, so $(j, \delta(m,b,v_1),\delta(m,b,v_2)) \in
\nsetc{\stype{\primt_2}{\primt_2)}}$.
If $U \neq \primt_1$, then the proof obligation is $(j, \delta(m,b,v_1),\delta(m,b,v_2) \in \nsetc{\stype{\primt_2}{\top}}$; this is trivial because any two values are related at $\top$.
For the case of standard signatures, at some point we have $\stype{\primt_1}{U_1}-> \stype{\primt_2}{U_2}$ and
we have to prove that $(j, \delta(m,b,v_1),\delta(m,b,v_2)) \in \nsetc{\stype{\primt_2}{\rpolicy(\stype{\primt_1}{U_1},\primt_2)}}$. Since
the signature is sound, we know that either $U_1 = \primt_1$ or $U_2 = \top$; then the result follows similarly to the primitive signature case.
\end{proof}

Note that the two syntactic principles for sound signatures of primitive types introduced in Section~\ref{sec:gobsec-primitive-signatures} are justified by the need to establish Lemma~\ref{lm:syntactic-implies-semantic}.
Principle (P1) is necessary because we cannot assume anything about two invocations of an {\em arbitrary} partial function $\delta$, except that given {\em syntactically} equal arguments, if it produces results, then those results are syntactically equal. Principle (P2) is justified 
because any two invocations of $\delta$ are observationally equivalent at $\top$ (like any computation in general), so the actual relation between the arguments does not matter. For any primitive operator signature that does not abide to either (P1) or (P2), it is possible to devise a $\delta$ that violates Lemma~\ref{lm:syntactic-implies-semantic}

Consequently, \gobsec with primitive types is a sound security-typed language, \ie~all well-typed programs satisfy PRNI (Theorem~\ref{the:gsound}).

\subsection{Illustration}

In Section~\ref{sec:gobsec-overview} we gave informal examples of secure and insecure programs with respect to PRNI. Now, armed with Theorem~\ref{the:gsound}, and the definitions for primitive types, we can formally check if a given program is secure by typechecking it. The prototype implementation of \gobsec features a number of examples and allows one to try out the language and typechecker.
In this section, we unfold the reasoning underlying the proof of Theorem~\ref{the:gsound} for a specific example, in order to illustrate the technical details of PRNI and the relational interpretation of object types, including primitive signatures.

To alleviate notation, we omit unused type parameters in method declarations and type instantiations in method invocations. Note that we introduce the $\Unit$ type with its unique $\unit$ primitive value. 

We illustrate polymorphic declassification by considering type and variable environments $\DeltaX \defas X:\String..\StringLen$ and $\Gamma \defas x: \stype{\String}{X}$. We discuss two possible formal definitions for $\StringLen$, either using standard method signatures, or using primitive signatures:
\begin{enumerate}
	\item $\rtypex{\alpha}{\left[\mathsf{length}: 
	\mtype{\Unit_{\lbot}}{\Int_{\lbot}\right]}}$
	\item $\rtypex{\alpha}{\left[\mathsf{length}: \stype{\Unit}{\ilab} -> \stype{\Int}{\ilab}\right]}$
\end{enumerate}

With definition 1) above, the program $\minv{x}{\mathsf{length}}{\unit}$ has type $\stypebot{\Int}$; \ie $\Delta;\Gamma |- \minv{x}{\mathsf{length}}{\unit} : \stypebot{\Int}$. 
Then, by Theorem~\ref{the:gsound}, we know that 
$\gtrni{\DeltaX}{\Gamma}{\minv{x}{\mathsf{length}}{\unit}}{\stypebot{\Int}}$ holds; the program is secure for any public observer.

Let us unfold $\gtrni{\DeltaX}{\Gamma}{\minv{x}{\mathsf{length}}{\unit}}{\stypebot{\Int}}$ to verify why it holds. 
For any type substitution $X \mapsto T \in \nsetd{\Delta}$ and equivalent value substitutions 
$(k,x \mapsto v_1,x \mapsto v_2) \in \nsetg{\bigcdot,x:\stype{\String}{T}}$, we have that
$(k,\minv{v_1}{\mathsf{length}}{\unit}),\minv{v_2}{\mathsf{length}}{\unit}) \in \nsetc{\stype{\Int}{\Int}}$.

To verify this: 
\begin{enumerate}
	\item By $(k,x \mapsto v_1,x \mapsto v_2) \in \nsetg{\bigcdot,x:\stype{\String}{T}}$ we know that 
$(k,v_1,v_2) \in \nsetv{\stype{\String}{T}}$. Because $T <: \StringLen$, we have $\nsetv{\stype{\String}{T}} \subseteq \nsetv{\stype{\String}{\StringLen}}$ by a subtyping lemma, and hence $(k,v_1,v_2) \in \nsetv{\stype{\String}{\StringLen}}$.

	\item Then, instantiate the definition of $\nsetv{\stype{\String}{\StringLen}}$ with $\mathsf{length},k,T,\unit,\unit$. Note that:
		\begin{itemize}
			\item $\mathsf{length} \in \StringLen$
			\item $\tlookup{\bigcdot,\StringLen}{\mathsf{length}}{\mtype{\Unit_{\lbot}}{\Int_{\lbot}}}$
			\item $T \in \String..\StringLen$, which follows from $X \mapsto T \in \nsetd{\Delta}$
			\item and $(k,\unit,\unit) \in \nsetv{\stype{\Unit}{\Unit}}$ (by definition of $\nsetv{\stype{\primt}{\primt}}$), 
		\end{itemize} 
\end{enumerate}
Then $(k,\minv{v_1}{\mathsf{length}}{\unit},\minv{v_2}{\mathsf{length}}{\unit}) \in \nsetc{\stype{\Int}{\Int}}$.

With definition 2) above of $\StringLen$, we apply the same steps until the instantiation of 
$\nsetv{\stype{\String}{\StringLen}}$. At this point, since $\mathsf{length}$ has a primitive signature, we
have to consider the extended case for primitive type signatures from Figure~\ref{fig:gobsec-pt-logical-relation}. Instantiate it with 
$\mathsf{length},k,\unit,\unit$, and observe that 
$\tlookup{\bigcdot,\StringLen}{\mathsf{length}}{\stype{\Unit}{\ilab} -> \stype{\Int}{\ilab}}$. Then, 
given that $(k,\unit,\unit) \in \nsetv{\stype{\Unit}{\Unit}}$, 
we have that $(k,\minv{v_1}{\mathsf{length}}{\unit},\minv{v_2}{\mathsf{length}}{\unit}) \in \nsetc{\stype{\Int}{\rpolicy(\stype{\Unit}{\Unit},\Int)}} = 
\nsetc{\stype{\Int}{\Int}}$.

\section{Related work}
\label{sec:gobsec-related-work}

{\bf Declassification.}
Expressive declassification policies were introduced by \citet{liZdancewic:popl2005} with the labels-as-functions approach. They
define two kinds of declassification policies: local and global. 
Local policies are concerned with one secret, while global policies express coordinated declassification of several secrets.
Label operations in this approach rely on a
semantic interpretation of declassification functions based on a general notion of program equivalence. In addition to the induced complexity, this precludes recursive policies.

The labels-as-types approach~\cite{cruzAl:ecoop2017} uses type interfaces instead of functions to express declassification policies. This simplifies the concepts involved (label ordering is simply subtyping), making an implementation more easily realizable. The approach naturally support local policies. More advanced typing disciplines such as refinement types~\cite{liquid:popl2008} could in principle be used to express global policies.

Conceptually, \citet{cruzAl:ecoop2017} relate
secure information flow with type abstraction, a connection also explored under different angles by \citet{bowmanAhmed:icfp2015} and \citet{washburnWeirich:lics2005}. \citet{bowmanAhmed:icfp2015} translate the noninterference result of the Dependency Core Calculus (DCC)~\cite{abadiAl:popl99} to parametricity, while \citet{washburnWeirich:lics2005} use 
information control mechanisms to ensure a generalized from of parametricity in presence of runtime type inspection.

The Decentralized Label Model (DLM)~\cite{myersLiskov:tosem2000} of Jif
enforces {\em robust declassification}~\cite{zdancewicMyers:csf2001}: restricting {\em who} can declassify values, using the {\em integrity policy} to ensure that the declassification is not triggered by conditions affected by an active attacker. Here, we do not model any notion 
of authority, focusing on the {\em what} dimension of declassification~\cite{sabelfeldSands:jcs2009}.

As noted by \citet{sabelfeldSands:jcs2009}, many declassification approaches of the {\em what} dimension can be expressed using partial equivalence relations to model the public observer knowledge. Here, we use the logical interpretation of security types (Figure~\ref{fig:gobsec-logical-relation})  
to specify the partial equivalence relation that a public observer can use to distinguish values and computations. 

{\bf Label polymorphism.}
Support for label polymorphism in security-typed programming languages can be classified in two categories: static and dynamic label polymorphism. 
Static label polymorphism can either be provided via explicit syntactic constructions to introduce generic labels~\cite{myers:jif}, or implicitly 
with constraint-based label inference~\cite{myers:jif,pottierSimonet:toplas2003,sunAl:sas2004}. The dynamic form of label polymorphism relies on labels as first-class entities that can be passed around like standard values~\cite{myers:jif,stefanAl:haskell2011,stefanAl:jfp2017}.

The Jif language~\cite{myers:jif} supports all three forms of label polymorphism. It provides a direct syntax to introduce labels at the method and class levels, which can be constrained. Also, Jif features label inference: local variables are inferred to have a fresh generic label 
that is resolved using constraints from the context. Inferred fields and method arguments have default labels. In addition, Jif
supports first-class labels. Our work focuses on 
the foundations of explicit declassification polymorphism, and currently does not address label inference and first-class labels. Because labels are types, label inference would boil down to fairly standard type inference; first-class labels however would require a notion of first-class types, which should be considered with care.

\citet{sunAl:sas2004} design a constraint-based label inference mechanism for an object-oriented language 
with classes and inheritance. Classes and methods are label polymorphic. The programmer can rely on the inference mechanism to achieve label polymorphism or to specify generic labels at the class level; method-level explicit
polymorphism is not considered. \citet{stefanAl:haskell2011, stefanAl:jfp2017} provide label-polymorphism via first-class labels much like Jif.

{\bf Declassification and Polymorphism.}
When present, the declassification mechanisms of the label-polymorphic proposals discussed above~\cite{myers:jif,pottierSimonet:toplas2003, sunAl:sas2004,stefanAl:jfp2017} are completely orthogonal to label polymorphism. The polymorphic labels-as-types approach developed in this work 
allows us to reason about declassification and label polymorphism with the single and unified concept of standard types.

Our approach is closely related to that of~\citet{hicksAl:plas2006}, which propose {\em trusted declassification} in an object-oriented language based on 
the DLM~\cite{myersLiskov:tosem2000}, where each label is composed of principals. Declassification is globally
defined, associating principals to the {\em trusted methods} that can be used to declassify an expression to another principal. 
Because classes are polymorphic with respect to principals, this induces a form of implicit label polymorphism.
More precisely, a class definition is checked at instantiation-time with the actual principals provided for the instantiation. This use-site polymorphism for principals is similar to our treatment of polymorphic primitive signatures (Section~\ref{sec:gobsec-primitive-types}).

\citet{tseZdancewic:esop2005} propose certificate-based declassification and conditioned noninterference. They extend System~$\fsub$ with monadic labels similarly to DCC~\cite{abadiAl:popl99}, using DLM~\cite{myersLiskov:tosem2000}. 
Declassification is modeled as a  
{\em read privilege} that a principal is allowed to give
to another principal in a certain context. Their work merges standard types with labels, principals and privileges in
the same syntactic category of types. Since System~$\fsub$ supports type polymorphism, the language supports label polymorphism. However, it is not clear how to use label polymorphism to express polymorphic declassification in that setting.

Finally, the syntactic principles we introduce for primitive signatures 
are related to the work of~\citet{liZdancewic:popl2005} on labels-as-functions. 
For local policies, the typing rules for integer primitive operators 
follow the same principles, but are more expressive. In particular, they provide typing rules for binary integer operators where one operand 
has an arbitrary declassification policy and the other operand is public; the resulting label is a functional composition of the operand label with a function wrapping the operator. As explained before, this semantic implication cannot be expressed with the labels-as-types approach, unless one is willing to consider much more advanced typing disciplines.

\section{Conclusion}
\label{sec:conclusion}

We extend relaxed noninterference in a labels-as-types approach to selective and expressive declassification in order to account for polymorphism. The proposed declassification polymorphism is novel and useful to precisely control declassification of polymorphic structures and to define procedures that are polymorphic over the declassification policies of their arguments.
Bounded polymorphism further controls the guarantees and expectations of clients and providers with respect to declassification.
Bringing type-based declassification to real-world programming also requires addressing the issue of primitive types, which were ignored in prior work. For this we introduce a novel form of ad-hoc polymorphism. We formalize the approach, prove its soundness, and provide a prototype implementation. 

This work provides a necessary and solid basis to integrate type-based declassification in existing languages. A particularly appealing alternative is to study the realization of our approach in Scala: its type system is expressive enough to encode bounded polymorphic declassification, and adjusting the typechecker to account for security levels (\ie the additional rule for method invocation) should be achievable via a compiler plugin.

\section*{Acknowledgment}
We thank C\u{a}t\u{a}lin Hri\c{t}cu and the anonymous reviewers for their detailed comments and suggestions.

\clearpage

\bibliographystyle{IEEEtranN}
\bibliography{../_bib/strings,../_bib/pleiad,../_bib/bib,../_bib/common}

\appendices
\newif\ifproofs
\proofstrue

\onecolumn

\section{Full syntax and semantics}

\subsection{Syntax}

\begin{figure}[htbp]
  \begin{small}		
		\begin{displaymath}
			\begin{array}{rcll}
				e & ::= & v |  e.m(e) | \Gbox{\gminv{e}{m}{U}{e}}  | x | \Gbox{\primb} & \text{(terms)}\\
				o & ::= & \gobject  & \text{(objects)}\\
				v & ::= & o | \Gbox{\primb} & \text{(values)}\\					
				T & ::= & O | \alpha | \Gbox{\primt} &  \text{(types)}\\
				U,A,B & ::= & T | \Gbox{X} & \text{(policies)}\\
				O & ::= & \rtypex{\alpha}{\gntrecordx{m}{M}} & \text{(object type)}\\
				R & ::= & \gntrecordx{m}{M} & \text{(record types)}\\
				M & ::= & \gmtype{\Gbox{X:A..B}}{S}{S} | \Gbox{I} & \text{(method signature)}\\
				I & ::= & \stype{\primt}{\ilab} -> \stype{\primt}{\ilab} &  \text{(primitive signatures)}\\								
				S & ::= & \stype{T}{U} | \Gbox{\stype{\primt}{\ilab}} &  \text{(security types)}\\				
				\primt & ::= & \text{(eg. Int)} & \text{(primitive type)}\\
			\end{array}		
		\end{displaymath}
 \end{small}
 \caption{\gobsec: Full Syntax}
  \label{fig:gobsec-syntax-appendix}
\end{figure}

\clearpage

\subsection{Environments}

\begin{small}
  \begin{displaymath}
    \begin{array}{rcll}            		  
			\Gamma & ::= & \bigcdot | \Gamma, x:S & \text{(type environment)} \\
			\Gbox{\DeltaX}	 & ::= & \Gbox{ \bigcdot | \DeltaX, X : A..B} & \text{(generic type variable environment)}\\
			\GammaSub & ::= & \bigcdot | \GammaSub, \alpha <: \beta | \GammaSub, \Gbox{\primt <: \beta} & \text{(subtyping environment)}\\
			\GammaTOk	 & ::= & \bigcdot | \GammaTOk, \alpha | \GammaTOk, X : A..B  & \text{(type variable environment)}\\
			\GammaST	 & ::= & \bigcdot | \GammaST, \alpha \triangleq O & \text{(type definition environment)}\\
    \end{array}		
  \end{displaymath}
 \end{small}

The type variable environment $\GammaTOk$ keeps track of the type variable definitions, both self type variables and generic 
type variables. This environment is not directly
used in the static semantics presented in the body of the paper. Recall that in the typing judgment $\DeltaX;\Gamma |- e: S$ of 
Figure~\ref{fig:gobsec-static-semantics}, we extend the typing environment $\Gamma$ with closed types regarding self type 
variables. For this reason we use $\DeltaX$ instead of $\GammaTOk$. Note that $\DeltaX$ just keeps track of generic type
variable definition. Then, the environment  $\GammaTOk$ is just an auxiliary mechanism to verify that an object type
is well-formed (Figure~\ref{fig:gobsec-ok-types})

\clearpage
\subsection{Type Equivalence}
\label{sec:gobsec-type-equivalence-appendix}

The type equivalence judgment (Figure~\ref{fig:gobsec-type-equivalence}) is essentially the same that the work of \citet{cruzAl:ecoop2017}, with minor 
modifications to the rule (O-congr) to support type variables.

Two types are equivalent (Figure~\ref{fig:gobsec-type-equivalence}) if the equivalence can be derived through the congruence induced by rules (Alpha-Eq) and (Fold-Unfold). For example:\\
 $\rtypex{\alpha}{\left[m: \left\langle  X: \alpha .. \top \right\rangle \alpha \rightarrow \alpha \right]} \equiv 
 {\rtypex{\beta}{\left[m: \left\langle  X: \beta .. \top \right\rangle \beta \rightarrow \beta \right]}}$ \\
 ${\rtypex{\alpha}{\left[m: \left\langle  X: \alpha .. \top \right\rangle \top  \rightarrow \alpha \right]}} \equiv
{\rtypex{\alpha}{\left[m: \left\langle  X: \alpha .. \top \right\rangle \top \rightarrow \rtypex{\beta}{\left[m: \left\langle  Y: \beta .. \top \right\rangle \top  \rightarrow \beta \right]} \right]}}$

\noindent

\begin{figure}[htbp]
\begin{small}
\framebox{$\typeeq{U}{U}$}
\begin{mathpar}
	\inference[(Sym)]{
	~
	}{	
		\typeeq{U}{U}
	}
	\quad
	\inference[(Refl)]{
		\typeeq{U_1}{U_2}
	}{
		\typeeq{U_2}{U_1}
	}
	\quad
	\inference[(Trans)]{
		\typeeq{U_1}{U_2} & \typeeq{U_2}{U_3}
	}{
		\typeeq{U_1}{U_3}
	}
	\\	
	\inference[(O-Congr)]{
		\typeeq{L_{i}}{L'_{i}} & \typeeq{U_{i}}{U'_{i}} &
		\typeeq{S_{1i}}{S'_{1i}} & \typeeq{S_{2i}}{S'_{2i}}
	}{
		\typeeq{\rtypex{\alpha}{\gntrecordx{m}{\gmtype{X:A..B}{S_1}{S_2}}}}{\rtypex{\alpha}{\gntrecordx{m}{\gmtype{X:L'..U'}{S'_1}{S'_2}}}}
	}
	\\
	\inference[(Alpha-Eq)]{
		O \triangleq \rtypex{\alpha}{\gntrecordx{m}{M}} &
		\beta~\text{fresh}
	}{
		\typeeq{O}{\ssubst{O}{\beta}{\alpha}}
	}
	\quad
	\inference[(Fold-Unfold)]{
		~
	}{
		\typeeq{O}{\ssubst{O}{O}{\alpha}}
	}
\end{mathpar}

\noindent
\framebox{$\stypeeq{S}{S}$}
\begin{mathpar}
	\inference{
		\typeeq{T_1}{T_2} & \typeeq{U_1}{U_2}
	}{
		\stypeeq{\stype{T_1}{U_1}}{\stype{T_2}{U_2}}
	}
\end{mathpar}
\end{small}
\caption{\gobsec: Type equivalence}%
\label{fig:gobsec-type-equivalence}%
\end{figure}

\clearpage
\subsection{Subtyping}

\begin{figure}[htbp]
\begin{small}
\framebox{$\subtp{\DeltaX; \GammaSub}{U}{U}$}
 \begin{mathpar}    
    \inference[(SPObj)]{
      \methods(\primt) = R_1 \quad O \triangleq \rtypex{\beta}{R_2} \\			
			\subtp{\DeltaX;\GammaSub, \primt <: \beta}{R_1}{R_2} \\
    }{
       \subtp{\DeltaX;\GammaSub}{\primt}{O}
    }\quad
    \inference[(SPVar)]{
      \primt <: \beta \in \GammaSub
    }{
      \subtp{\DeltaX;\GammaSub}{\primt}{\beta}
    }
		\\
    \inference[(SObj)]{
      O_1 \triangleq  \rtypex{\alpha}{R_1} \quad O_2 \triangleq \rtypex{\beta}{R_2} \\
			\subtp{\DeltaX; \GammaSub, \alpha <: \beta}{R_1}{R_2}
    }{
       \subtp{\DeltaX;\GammaSub}{O_1}{O_2}
    } \quad
		\inference[(SGVar1)]{
      X : A..B \in \DeltaX
    }{
      \subtp{\DeltaX;\GammaSub}{X}{B}
    }
		\quad
		\inference[(SGVar2)]{
      X : A..B \in \DeltaX
    }{
      \subtp{\DeltaX;\GammaSub}{A}{X}
    }
		\\
		\inference[(SVar)]{
      \alpha <: \beta \in \GammaSub
    }{
      \subtp{\DeltaX;\GammaSub}{\alpha}{\beta}
    } \quad
    \inference[(SSubEq)]{
      \typeeq{O_1}{O_2}
    }{
      \subtp{\DeltaX;\GammaSub}{O_1}{O_2}
    }
    \\
    \inference[(STrans)]{
      \subtp{\DeltaX;\GammaSub}{U_1}{U_2} & \subtp{\DeltaX;\GammaSub}{U_2}{U_3}
    }{
      \subtp{\DeltaX;\GammaSub}{U_1}{U_3}
    }
\end{mathpar}
\framebox{$\subtps{\DeltaX;\GammaSub}{R_1}{R_2}$}
\begin{mathpar}
		\inference[(SR)]{
				\overline{m'} \subseteq  \overline{m} \quad m_{i} = m'_{j} \implies \subtp{\DeltaX; \GammaSub}{M}{M'}
			}{
				 \subtp{\DeltaX;\GammaSub}{\gntrecordx{m}{M}}{\gntrecordx{m'}{M'}}
			}
\end{mathpar}
\framebox{$\subtps{\DeltaX;\GammaSub}{M}{M}$}
\begin{mathpar}
		\inference[(SM)]{								
				\subtp{\DeltaX; \GammaSub}{B'}{B} \quad \subtp{\DeltaX; \GammaSub}{A}{A'} \\
				\subtp{\DeltaX,X:A'..B';\GammaSub}{S'_1}{S_1} &
				\subtp{\DeltaX,X:A'..B';\GammaSub}{S_2}{S'_2}
			}{
				 \subtp{\DeltaX;\GammaSub}{\gmtype{X:A..B}{S_{1}}{S_{2}}}{\gmtype{X:A'..B'}{S'_{1}}{S'_{2}}}
			} \quad		
		\inference[(SImpl)]{
    }{
       \subtp{\DeltaX;\GammaSub}{I}{I}
    }
\end{mathpar}
\framebox{$\subtps{\DeltaX;\GammaSub}{S}{S}$}
 \begin{mathpar}
   \inference[(TSubST)]{
     \subtp{\DeltaX;\GammaSub}{T_1}{T_2} &    
     \subtp{\DeltaX;\GammaSub}{U_1}{U_2} &
   }{
     \subtps{\DeltaX;\GammaSub}{\stype{T_1}{U_1}}{\stype{T_2}{U_2}}
   } 
  \end{mathpar}
\end{small}
\caption{\gobsec: Subtyping rules} 
\label{fig:gobsec-subtyping-appendix}
\end{figure}

\clearpage
\subsection{Standard well formedness of types and environments}
\label{sec:gobsec-well-formedness-rules-appendix}

The Fig.~\ref{fig:gobsec-ok-env} and Fig.~\ref{fig:gobsec-ok-types} show the rules for well formedness of environments and 
types respectively. 

Since $\DeltaX$ match the syntax of $\GammaTOk$, we can use the judgments 
$\wft{\GammaTOk}{U}$ and $\wft{\GammaTOk}{S}$ with $\DeltaX$ and 
they hold if $U$ (an $S$ respectively) are closed with respecto to self type variables $\alpha$.

\begin{figure}[h]
\framebox{$\wfex{\DeltaX}{\Gamma}$}
\begin{mathpar}
	\inference{
		~
	}{
		\wfex{\DeltaX}{\cdot}
  }\quad
	\inference{
		\wfex{\DeltaX}{\Gamma} &  \wf{\DeltaX}{S} & x \notin dom(\Gamma)
	}{
		\wfex{\DeltaX}{\Gamma,x:S}
}
\end{mathpar}

\framebox{$\wfe{\DeltaX}$}
\begin{mathpar}	
	\inference{
		~
	}{
		\wfe{\cdot}
	}
	\quad
	\inference{
		\wfe{\DeltaX} & X \notin dom(\DeltaX) & \DeltaX |- A & \DeltaX |- B &
	}{
		\wfe{\DeltaX, X: A..B}
	}	
\end{mathpar}

\framebox{$\wfe{\GammaSub}$}
\begin{mathpar}	
	\inference{
		~
	}{
		\wfe{\cdot}
	}
	\quad
	\inference{
		\wfe{\GammaSub} & \alpha_i \notin dom(\GammaSub) \cup cod(\GammaSub)
	}{
		\wfe{\GammaSub, \alpha_1 <: \alpha_2}
	}
	\quad
	\inference{
		\wfe{\GammaSub} & \alpha_2 \notin dom(\GammaSub) \cup cod(\GammaSub) &  \ptsound{\primt}
	}{
		\wfe{\GammaSub, \primt <: \alpha_2}
	}
\end{mathpar}

\caption{\gobsec: Environments well-formedness}%
\label{fig:gobsec-ok-env}%
\end{figure}

\begin{figure}[h]
\begin{small}
	\framebox{$\wft{\GammaTOk}{U}$}
	\begin{mathpar}
		\inference{
			\ptsound{\primt}
		}{
		 \wft{\GammaTOk}{\primt}
		}\quad
		\inference{
			X \in \GammaTOk
		}{
		 \wft{\GammaTOk}{X}
		} \quad
		\inference{
			\alpha \in \GammaTOk
		}{
		 \wft{\GammaTOk}{\alpha}
		} 
		\quad
		\inference{
			O \equiv \rtypex{\alpha}{\gntrecordx{m}{\gmtype{X : A..B}{S_1}{S_2}}} \\
			(i \neq j \implies m_i \neq m_j) \\		
		\wft{\GammaTOk,\alpha}{A} & 
		\wft{\GammaTOk,\alpha}{B} \\
		\wft{\GammaTOk,\alpha,X: A..B}{S_{1i}} & 
		\wft{\GammaTOk,\alpha,X: A..B}{S_{2i}}
		}	{
			\wft{\GammaTOk}{O}
		}
	\end{mathpar}	
\end{small}
\vspace{5mm}

\begin{small}	
	\framebox{$\wft{\GammaTOk}{S}$}	
	\begin{mathpar}	
			\inference{
				\wft{\GammaTOk}{T} \quad \wft{\GammaTOk}{U}
			}{
				\wft{\GammaTOk}{\stype{T}{U}} 
			}
	\end{mathpar}
\end{small}

\caption{\gobsec: Well-formedness of types}
\label{fig:gobsec-ok-types}
\end{figure}

\clearpage
\subsection{Well-formedness of security types facet-wise}

\begin{displaymath}
		\begin{array}{lcl}
			\mathsf{soundsig}(\gmtype{\_}{\stype{\primt_1}{U_1}}{\stype{T_2}{U_2}}) & <=> & \primt_1 = U_1~\vee~U_2 = \top
		\end{array}
\end{displaymath}

Fig.~\ref{fig:gobsec-x-appendix} shows extended subtyping judgments that are the same the ones of Figure~\ref{fig:gobsec-subtyping-appendix}, except
that the judgment $\wfsub{\DeltaX;\GammaSub}{M}{M}$ adds the rule (IG) justifying that is OK to declassify a primitive signature with a concrete sound signature.

The judgment $\wf{\DeltaX}{S}$ holds if the type $S$ is a closed type and a well-formed security type. 

Having formalized well-formedness of environments and types, we assume them in most definitions.

\begin{figure}[htbp]
\begin{small}
\framebox{$\wfsub{\DeltaX; \GammaSub}{U}{U}$}
\begin{mathpar}
 \cdots
\end{mathpar}
\framebox{$\wfsub{\DeltaX;\GammaSub}{R_1}{R_2}$}
\begin{mathpar}
 \cdots
\end{mathpar}
\framebox{$\wfsub{\DeltaX;\GammaSub}{M}{M}$}
\begin{mathpar}
		\cdots \quad
		\inference[(IG)]{
			\wfsub{\DeltaX;\GammaSub}{T_1}{\primt_1} &  \wfsub{\DeltaX;\GammaSub}{\primt_2}{T_2} & \mathsf{soundsig}(\gmtype{\_}{\stype{T_1}{U_1}}{\stype{T_2}{U_2}})
		}{
		 	\wfsub{\DeltaX;\GammaSub}{\stype{\primt_1}{*} -> \stype{\primt_2}{*}}{\gmtype{\_}{\stype{T_1}{U_1}}{\stype{T_2}{U_2}}} 
		}
\end{mathpar}
\framebox{$\wfsub{\DeltaX;\GammaSub}{S}{S}$}
\begin{mathpar}
 \cdots
\end{mathpar}
\end{small}
\caption{\gobsec: Rules for valid relation between facets} 
\label{fig:gobsec-x-appendix}
\end{figure}

\begin{figure}[h]
\begin{small}
	\framebox{$\wfft{\DeltaX;\GammaST}{U}$}
	\begin{mathpar}	
			\inference{
				O \equiv \rtypex{\alpha}{\gntrecordx{m}{M}} \\
				\wfft{\DeltaX;\GammaST, \alpha:O}{M} 
				}{
				 \wfft{\DeltaX;\GammaST}{O}
			}\quad
			\inference{
				 ~
			}{
			 \wfft{\DeltaX;\GammaST}{\alpha}
		 }\quad
			\inference{
				 ~
			}{
			 \wfft{\DeltaX;\GammaST}{X}
		 }
		\quad
			\inference{
				 ~
			}{
			 \wfft{\DeltaX;\GammaST}{\primt}
		 }
	\end{mathpar}
	\framebox{$\wfft{\DeltaX;\GammaST}{M}$}
	\begin{mathpar}	
			\inference{				
				\wfft{\DeltaX,X:A..B;\GammaST}{S_{1}} & \wfft{\DeltaX,X:A..B;\GammaST}{S_{2}}
				}{
				 \wfft{\DeltaX;\GammaST}{\gmtype{X:A..B}{S_1}{S_2}}
			} \quad
			\inference{				
				~
				}{
				 \wfft{\DeltaX;\GammaST}{I}
			}
	\end{mathpar}
\begin{mathpar}
\end{mathpar}

	\framebox{$\wfft{\DeltaX;\GammaST}{S}$}
	\begin{mathpar}	
		\inference{
			\wfft{\DeltaX;\GammaST}{T} \quad \wfft{\DeltaX;\GammaST}{U}  \quad \wfsub{\DeltaX;\cdot}{\unfold{T}}{\unfold{U}}
		}{
			\wfft{\DeltaX;\GammaST}{\stype{T}{U}} 
		
		}
\end{mathpar}


	\framebox{$\wf{\DeltaX}{S}$}
	\begin{mathpar}	
		\inference{
			\wft{\DeltaX}{S} ~~\wfft{\DeltaX;\cdot}{S}
			}{
			 \wf{\DeltaX}{S}
		}
	\end{mathpar}

\end{small}
 \caption{\gobsec: Well-formedness of security types facet-wise}
  \label{fig:gbosec-ok-types-subtyping}
\end{figure}

\clearpage
\subsection{Auxiliary definitions}
\label{sec:gobsec-auxiliary-definition-appendix}

\begin{figure}[htbp]
	\begin{small}
		\framebox{$\ubound{\DeltaX}{U} = T$}
			\begin{mathpar}
				 \inference{
					T \neq X
				 }{
					 \ubound{\DeltaX}{T} = T
				 }
				\quad 
				\inference{
					X: A..B \in \DeltaX \\
					\ubound{\DeltaX}{B} = T_2
				}{
					\ubound{\DeltaX}{X} = T_2
				}			
			\end{mathpar}
		\framebox{$\DeltaX |- m \in U$}
		\begin{mathpar}
			\inference{		
				O \triangleq  \rtypex{\alpha}{\gntrecordx{m}{M}} 
			 }{
				\DeltaX |-m_i \in O
			 }\quad
			 \inference{		
				\methods(\primt) = \gntrecordx{m}{I}     
			 }{
				 \DeltaX |- m_i \in \primt
			 }\quad
			 \inference{		
				\DeltaX |- m \in \ubound{\DeltaX}{X}
			 }{
				 \DeltaX |- m \in X
			 }	
		\end{mathpar}
		\framebox{$\tlookup{\DeltaX,U}{m}{M}$}
		\begin{mathpar}
				\inference{
					O \triangleq  \rtypex{\alpha}{\gntrecordx{m}{M_i}}
			 }{
				 \gtlookup{\_}{O}{m_i} = \ssubst{M_i}{O}{\alpha}
			 } \quad
			 \inference{
				\methods(\primt) =  \gntrecordx{m}{I}
			 }{
				 \gtlookup{\_}{\primt}{m_i} = I_i
			 }
			\\
			 \inference{
					~
				 }{
					 \gtlookup{\DeltaX}{X}{m} = \gtlookup{\_}{\ubound{\DeltaX}{X}}{m}
				 }	
		\end{mathpar}		
		\framebox{$\methimpl{o}{m}{x.e}$}
		\begin{mathpar}
			 \inference{		
				o \triangleq \objectx{z}{S}
			 }{
				 \methimpl{o}{m_i}{x.e_i}
			 }	
		\end{mathpar}		
		\framebox{$\DeltaX |- U \in A..B$}
		\begin{mathpar}			
			\quad		
			 \inference{		
				\subtp{\DeltaX;\bigcdot}{A}{U}  \quad \subtp{\DeltaX;\bigcdot}{U}{B}
			 }{
				 \DeltaX |- U \in A..B
			 }	
		\end{mathpar}
		\framebox{$\rpolicy(\stype{\primt}{U},\primt) = U$}
	\begin{mathpar}
		\begin{array}{rl}
			\rpolicy(\stype{\primt_1}{U_1},\primt_2) = & 
			\begin{cases}
							\primt_2 & \primt_1 = U_1\\
							\top &	\textit{otherwise} \\
			\end{cases}			
		\end{array}
	\end{mathpar}
	\end{small}
\caption{\gobsec: Auxiliary definitions}
\label{fig:gobsec-aux-definitions-appendix}%
\end{figure}

\clearpage
\subsection{Typing}
\label{sec:gobsec-typing-appendix} 

\begin{figure}[htbp]
\begin{small}
\framebox{$\DeltaX; \Gamma |- e : S$}
\begin{mathpar}
	\inference[(TVar)]{
				x \in dom(\Gamma)
			}{
				\DeltaX; \Gamma |- x: \Gamma(x)
			} \quad
  \inference[(TSub)]{
			\DeltaX;\Gamma |- e: S' \quad \subtps{}{S'}{S}
    }{
      \DeltaX;\Gamma |- e: S
		}
\end{mathpar}
\begin{mathpar}
		\inference[(TPrim)]{			
			\primt = \Delta_{\primb}(\primb)
    }{
      \Gamma |- \primb: \stype{\primt}{\primt}
    } \quad
		\inference[(TObj)]{
			S \triangleq \stype{T}{U} \quad 
			\tlookup{\_,T}{m_i}{\gmtype{X:A_{i}..B_{i}}{S'_i}{S''_i}} \\
			\DeltaX,X:A_{i}..B_{i};\Gamma, z: S, x:S^{'}_i |- e_i : {S''_i}
    }{
      \DeltaX;\Gamma |- \objectx{z}{S}: S
    }
\end{mathpar}

\begin{mathpar}
  \inference[(TmD)]{
      \DeltaX;\Gamma |- e_1 : \stype{T}{U} \quad
			\DeltaX |- m \in U \quad
			\tlookup{\DeltaX;U}{m}{\gmtype{X:A..B}{S_1}{S_2}}	\\
			\DeltaX |- U' \in A..B &
			\DeltaX;\Gamma |- e_2 :  \ssubst{S_1}{U'}{X}
    }{
      \DeltaX;\Gamma |- \gminv{e_1}{m}{U'}{e_2}: \ssubst{S_2}{U'}{X}
    }		
\end{mathpar}
\begin{mathpar}
  \inference[(TmH)]{
      \DeltaX;\Gamma |- e_1 : \stype{T}{U} \quad
      \DeltaX |- m  \notin U \quad
			\gtlookup{\_}{T}{m}{\gmtype{X:A..B}{S_1}{\stype{T_2}{U_2}}} \\
			\DeltaX |- U' \in A..B
			&
			\DeltaX;\Gamma |- e_2 : \ssubst{S_1}{U'}{X}
    }{
      \DeltaX;\Gamma |- \gminv{e_1}{m}{U'}{e_2}: \stype{\ssubst{T_2}{U'}{X}}{\top}			
			}
			\\
		 \inference[(TPmD)]{
			\DeltaX;\Gamma |- e_1 : \stype{T}{U} & \DeltaX |- m \in U &
			\tlookup{\DeltaX,U}{m}{\stype{\primt_1}{\ilab} -> \stype{\primt_2}{\ilab}}	\\ 
			\DeltaX;\Gamma |- e_2 : \stype{\primt_1}{U_1} & 
			\rpolicy(\stype{\primt_1}{U_1},\primt_2) = \primt'_2
		}{
			\DeltaX;\Gamma |- \minv{e_1}{m}{e_2}: \stype{\primt_2}{\primt'_2}
		}\\
		\inference[(TPmH)]{
			\DeltaX;\Gamma |- e_1 : \stype{T}{U}  & \DeltaX |- m \notin U &
			\tlookup{\DeltaX,T}{m}{\stype{\primt_1}{\ilab} -> \stype{\primt_2}{\ilab}}	\\				
			\DeltaX;\Gamma |- e_2 : \stype{\primt_1}{U_1}
		}{
			\DeltaX;\Gamma |- \minv{e_1}{m}{e_2}: \stype{\primt_2}{\top}
		}
\end{mathpar}
\end{small}
\caption{\gobsec: Static semantics}
  \label{fig:gobsec-static-semantics-appendix}
\end{figure}

\clearpage
\subsection{Dynamic semantics}
\label{sec:gobsec-dynamic-semantics-appendix}

\begin{figure}[h]
\begin{small}
	\begin{mathpar}			
			\begin{array}{llll}
				E & ::= & \left[~\right] | E.m(e) |  v.m(E)& \text{(evaluation contexts)}
			\end{array}\\
      \inference[(EMInvO)]{
        o \triangleq \objectxx{z}{\_}{\_} \quad \methimpl{o}{m}{x.e}
      }{
        E[\gminv{o}{m}{\_}{v}] \reduce E[\mbsubst{e}{o}{z}{v}{x}]
      }			
	\end{mathpar}
	\begin{mathpar}	 
	  \inference[(EMInvP)]{
		  ~
		}{
			E[\primb_1.m(\primb_2)] \reduce E[\theta(m,\primb_1,\primb_2)]
		}
	\end{mathpar} 
\end{small}
\caption{\gobsec: Full Dynamic semantics }
\label{fig:gobsec-dynamic-semantics-appendix}
\end{figure}

\clearpage
\subsection{Simple Type System}
\label{sec:gobsec-simple-ts-appendix}

Figure~\ref{fig:gobsec-safe-type-system} and Figure \ref{fig:gobsec-safe-subtyping} define a simple type system and 
simple subtyping judgments respectively that do not take into account the declassification type. The typing judgment $\Gamma \safevdash e: S$ uses 
the subtyping judgment $\GammaSub \safevdash S <: S$. The definition of the judgment $\GammaSub \safevdash T <: T$ is straightforward and then omit here.

\begin{figure}[htbp]
\begin{small}
\framebox{$\Gamma \safevdash e: S$}
\begin{mathpar}
	\inference[(T1Var)]{
				x \in dom(\Gamma)
			}{
				\Gamma \safevdash x: \Gamma(x)
			} \quad
  \inference[(T1Sub)]{
			\Gamma \safevdash e: S' \quad \simplesub{}{S'}{S}
    }{
      \Gamma \safevdash e: S
		}
		 \quad
  \inference[(T1Prim)]{
			 \primt = \Theta(\primb)
			}{
			\Gamma \safevdash \primb: \stype{\primt}{\primt}
		}
\end{mathpar}
\begin{mathpar}
  \inference[(T1Obj)]{			
			\tlookup{\_,T}{m_i}{\gmtype{\_}{S'_i}{S''_i}} &   
			\Gamma, z: S, x_i:S^{'}_i \safevdash e_i : {S''_i}
    }{
      \Gamma \safevdash \objectx{z}{S}: S
    }
\end{mathpar}
\begin{mathpar}
  \inference[(T1mI)]{
      \Gamma \safevdash e_1 : \stype{T}{U} 
      &	\tlookup{\_,T}{m}{\gmtype{\_}{S_1}{S_2}} & 
			\Gamma \safevdash e_2 : S_1
    }{
      \Gamma \safevdash \gminv{e_1}{m}{\_}{e_2}: S_2
			}
\end{mathpar}
\begin{mathpar}
  \inference[(T1PmI)]{
      \Gamma \safevdash e_1 : \stype{T}{U} 
      &	\tlookup{\_,T}{m}{\stype{\primt_1}{\ilab} -> \stype{\primt_2}{\ilab}} & 
			\Gamma \safevdash e_2 : \stype{\primt_1}{\ilab}
    }{
      \Gamma \safevdash \minv{e_1}{m}{e_2}: \stype{\primt_2}{\ilab}
			}
\end{mathpar}

\framebox{$\stypeof{\Gamma}{e}{T}$}
\begin{mathpar}
  \inference{
    \Gamma \safevdash e : \stype{T}{U}
  }{
    \stypeof{\Gamma}{e}{T}
  }
\end{mathpar}
\end{small}
\caption{\gobsec simple typing, defined in terms of single-facet typing}%
\label{fig:gobsec-safe-type-system}%
\end{figure}

\begin{figure}[htbp]
\begin{small}
\framebox{$\GammaSub \safevdash T <: T$}
\begin{mathpar}
	\cdots
\end{mathpar}
\framebox{$\GammaSub \safevdash S <: S$}
\begin{mathpar}
	\inference{
    \GammaSub \safevdash T_1 <: U_1
  }{
    \GammaSub |-_{\mathsf{sf}} \stype{T_1}{U_1} <: \stype{T_2}{U_2}
  }
\end{mathpar}
\end{small}
\caption{\gobsec simple subtyping}%
\label{fig:gobsec-safe-subtyping}%
\end{figure}

\ifproofs

\clearpage
\section{Unary model}
\label{sec:gobsec-unary-logical-relation}

We present the logical predicate for safety (Figure~\ref{fig:gobsec-unary-logical-relation-appendix}). This logical predicate
gives a safety meaning to a safety type, hence the declassification type does not play any role in the definitions.

\begin{figure}[htbp]
  \begin{small}
		\begin{displaymath}
			\begin{array}{rcl}
			\setvx{k}{\primt} & = & 
			\begin{block}				
				\{ \primb |\\
				\quad (\forall j < k.\; v \in \setvx{j}{\primt} ~\wedge~ \\
				\qquad (\forall m \in \primt, \primb'. ~\tlookup{\bigcdot,\primt}{m}{\stype{\primt'}{\ilab} -> \stype{\primt''}{\ilab}} \\
				 \qquad ~~ \primb'\in \setvx{j}{\primt'} \implies \theta(m,\primb,\primb') \in \setcx{j}{\primt''})) \} 
			\end{block} \\
			& & \\
			\setvx{k}{O} & = & 
			\begin{block}				
				\{ v = \objectxx{z}{\stype{O_1}{\_}}{\_} ~|~  |-_1 \stype{O_1}{\_} <: \stype{O}{\_}~\wedge~ \\
				\quad (\forall j < k.\; v \in \setvx{j}{O_1} ~\wedge~ \\
				\qquad (\forall m \in O, v'. ~\tlookup{\bigcdot;O}{m}{\gmtype{\_}{\stype{T'}{\_}}{\stype{T''}{\_}}} \quad \methimpl{v}{m}{x.e} \\
				 \qquad ~~ v'\in \setvx{j}{T'} \implies \mbsubst{e}{v}{z}{v'}{x} \in \setcx{j}{T''})) \} 
			\end{block} \\
			& & \\
			\setcx{k}{T} & = & 
			\begin{block}
			\{ e | \forall j < k. ~\forall e'. (e  \reduce^j e' ~\wedge~ \mathsf{irred}(e')) \implies e' \in \setvx{k-j}{T} \} 
			\end{block}
			\end{array}
		\end{displaymath}
	 \end{small}
 \caption{Unary logical relation for safety}
  \label{fig:gobsec-unary-logical-relation-appendix}
\end{figure}

\clearpage
\section{Type safety}
\label{sec:gobsec-type-safety-appendix}

\gobsecSafe*

\begin{definition}[Semantic typing]
\label{def:gobsec-semantictyping}
$\models e : \stype{T}{U} \Longleftrightarrow  \forall k \ge 0.~e \in \setcx{k}{T}$.
\end{definition}

\begin{restatable}[Semantic type safety]{lemma}{typesafetyB}
\label{lm:gobsec-typesafetyB}
$\models e : S \implies \mathsf{safe}(e)$
\end{restatable}

\begin{proof}
The proof follows directly from definitions ~\ref{def:gobsec-safe} and \ref{def:gobsec-semantictyping} .
\end{proof}

\begin{restatable}[Type Safety]{lemma}{safetyTPrim}
\label{lm:gobsec-type-safety}
$\Gamma \vdash e : S \implies \Gamma \models e : S$
\end{restatable}
\begin{proof}
The proof is by induction on the typing derivation of $e$. The different case with respect to the proof of type safety 
of \obsec~\cite{cruzAl:ecoop2017} is the case (TPrim). For that case we only need to assume that 
$\Theta$ and $\theta$---which are parameters of the language---agree on the specification and implementation of all primitive types and their operations.
\end{proof}

\gobsecTypeSafety*
\begin{proof}
This follows directly from Lemmas~\ref{lm:gobsec-typesafetyB} and~\ref{lm:gobsec-type-safety}
\end{proof}

\clearpage
\section{Polymorphic relaxed noninterference}
\renewcommand{\rhosyn}[1]{\sigma(#1)} 
\renewcommand{\rhosynx}[2]{{#1}(#2)}
\renewcommand{\gsetv}[1]{\nsetv{#1}}
\renewcommand{\gsetc}[1]{\nsetc{#1}}
\renewcommand{\gsetd}[1]{\nsetd{#1}}
\renewcommand{\gsetg}[1]{\nsetg{#1}}

\subsection{Logical relation for PRNI}

Figure~\ref{fig:gobsec-lg-prni-forall-type} shows the full logical relation for PRNI

\begin{figure}[h]
	\begin{small}  
	\begin{displaymath}
		\begin{array}{lcl}
			\sigma  & ::= & \emptyset | \sigma[X \mapsto T] \\
			\atomone{n}{T} & = & \{(k,e_1,e_2) | k < n ~\wedge ~|-_1 e_1: T~\wedge~|-_1 e_2:T\} \\
			\atomvalone{n}{T} & = & \{(k,v_1,v_2) \in \atomone{n}{T} \} \\
			\atomunion{T} & = & \{ (k,e_1,e_2) \in \bigcup\limits_{n \ge 0} \atomone{n}{T} \}\\
			\nsetv{\stype{\primt}{\primt}} & = & \{(k,\primb,\primb) \in \atomunion{\primt}\} \\
			\nsetv{\stype{T}{O}}  &  = & 
			\begin{block}
				\{ (k, v_1, v_2) \in \atomunion{T} | 
				\\
				((\forall m \in O . \quad
					\tlookup{O}{m}{\gmtype{X:A..B}{S'}{S''}}					
					\\
					\quad \forall j < k,T', v_1',v_2'. \quad |- T' ~\wedge~ T' \in A..B ~\wedge~ \\
					\quad (j,v_1,v_2) \in \nsetv{\stype{T}{O}} ~\wedge~ (j, v_1', v_2') \in \nsetv{\ssubst{S'}{T'}{X}} \implies
					\\					 
					\quad \quad (j, \gminv{v_1}{m}{\_}{v'_1}, \gminv{v_2}{m}{\_}{v'_2}) \in \nsetc{\ssubst{S''}{T'}{X}} )~\wedge~ \\
				(\forall m \in O . \quad \tlookup{O}{m}{\stype{\primt_1}{\ilab} -> \stype{\primt_2}{\ilab}} \\
					\quad \forall j < k,v_1',v_2'. U_1 >: \primt_1\\
					\quad ((j,v'_1,v'_2) \in \nsetv{\stype{\primt_1}{U_1}} 
					\implies (j,v_1.m(v'_1),v_2.m(v'_2)) \in \nsetc{\stype{\primt_2}{\rpolicy(\stype{\primt_1}{U_1},\primt_2)}}))
				\}
			\end{block} \\			
			\nsetc{\stype{T}{U}} &  =  & 
			\begin{block}			 
				\{ (k, e_1, e_2) \in \atomunion{T}|
				(\forall j < k.
					 (e_1 \reduce^{\le j} v_1~\wedge~e_2 \reduce^{\le j} v_2 )
					\implies (k-j,v_1,v_2) \in \nsetv{\stype{T}{U}})
				\}
			\end{block} \\
			\nsetg{\cdot} & = & \{(k,\emptyset,\emptyset)\} \\
			\nsetg{\Gamma;x:S} & = &
				\begin{block}					
					\{(k,\gamma_1\left[x \mapsto v_1 \right],\gamma_2\left[x \mapsto v_2 \right])|
					(k,\gamma_1,\gamma_2) \in \nsetg{\Gamma}
					~\wedge
					~ (k,v_1, v_2) \in \nsetv{S}
					\} 
				\end{block} \\
			\nsetd{\cdot} & = & \{\emptyset\} \\
			\nsetd{\DeltaX;X:A..B} & = & 
			\begin{block}
				\{\sigma\left[X \mapsto T \right] |
					~\sigma \in \nsetd{\DeltaX} \wedge \DeltaX |- T \in A..B\} \\
			\end{block} \\
			\gtrni{\DeltaX}{\Gamma}{e}{S} & \Longleftrightarrow &
			\begin{block}
				\quad S \triangleq \stype{T}{U} \quad \stypeof{\Gamma}{e}{T}~\wedge~ \DeltaX |- \Gamma~\wedge~\DeltaX |- S\\
				\quad \forall k \ge 0 .~\forall \sigma, \gamma_1,\gamma_2.
				~\sigma \in \nsetd{\DeltaX}. 
				~ (k,\gamma_1,\gamma_2) \in \nsetg{\sigma(\Gamma)} 
				\\
				\quad \implies (k,\sigma(\gamma_1(e)),\sigma(\gamma_2(e))) \in \nsetc{\sigma(S)}
			\end{block}\\
			\DeltaX ;\Gamma |- e_1 \approx e_2 : S  & \Longleftrightarrow &
			\begin{block}
				\DeltaX;\Gamma |- e_i: S~\wedge \forall k \ge 0 .~\forall \sigma, \gamma_1,\gamma_2.
				~\sigma \in \gsetd{\DeltaX}. \\
				\quad \quad (k,\gamma_1, \gamma_2) \in \gsetg{\sigma(\Gamma)} \implies (k, \rhosyn{\gamma_1(e_1)}, \rhosyn{\gamma_2(e_2)} \in \nsetc{\sigma(S)}
			\end{block}
			\end{array}\\
  \end{displaymath}
	\end{small}
 \caption{\gobsec. Step-indexed logical relations for PRNI}  
 \label{fig:gobsec-lg-prni-forall-type}
\end{figure}

\begin{figure}[h]
	\begin{small}  
	\begin{displaymath}
		\begin{array}{rcl}			
			\gamma & ::= & \emptyset | \extgamma{\gamma}{x}{v} \\
			\emptyset \models \Gamma &  & \\
			\extgamma{\gamma}{x}{v} \models \Gamma,x:S & \iff & 
			\begin{block}
				\gamma \models \Gamma ~\wedge~  \bigcdot;\bigcdot |- v: S
			\end{block}
			\end{array}\\
  \end{displaymath}
	\end{small}
 \caption{\gobsec. PRNI. Auxiliary definition (used in proofs)}  
 \label{fig:gobsec-lg-auxiliary-definitions}
\end{figure}

\clearpage
\subsection{Auxiliary Lemma: Simple typing}

\begin{lemma}[Well-type programs are simple well-typed]
\label{lm:gobsec-security-ts-implies-simple-ts}
\mbox{}
\\
If $\DeltaX;\Gamma |- e: \stype{T}{U}$ then $\stypeof{\Gamma}{e}{T}$
\end{lemma}

\begin{lemma}[Security subtyping implies simple subtyping]
\label{lm:gobsec-security-st-implies-simple-st}
\mbox{}
\\
If $\DeltaX;\GammaSub |- T' <: T $ then $\GammaSub \safevdash T' <: T$
\end{lemma}

\begin{lemma}[Value substitution preserves simple typing]
\label{lm:gobsec-value-substitution-preserve-simple-typing}
\mbox{}

If $\stypeof{\Gamma}{e}{T}$ and $ \gamma \models \Gamma$ then ${\stypeof{}{\gamma(e)}{T}}$
\end{lemma}
\begin{proof}
By induction on $\Gamma$.
\end{proof}

\clearpage
\subsection{Auxiliary Lemma: Logical relation}

\subsubsection{Atom subtyping}
\begin{lemma}[Atom subtyping]
\label{lm:gobsec-atom-subtyping}
\mbox{}
\\
If $(k,v_1,v_2) \in \atomunion{T'}$ and  $\bigcdot;\bigcdot |- T' <: T$  \\
then $(k,v_1,v_2) \in \atomunion{T}$
\end{lemma}
\begin{proof}
Proof obligations:

$|-_1 e_i: T$ ($i \in \{1,2\})$. Apply rule (T1Sub). Note that $\bigcdot;\bigcdot |- T' <: T => \bigcdot \safevdash T'<:T$ 
(Lemma~\ref{lm:gobsec-security-st-implies-simple-st}, \nameref{lm:gobsec-security-st-implies-simple-st})
\end{proof}

\subsubsection{Atom reduction}
\begin{lemma}[Atom reduction]
\label{lm:gobsec-atom-reduction}
\mbox{}
\\
Let $e_1 \reduce^{*} e'_1$ and $e_2 \reduce^{*} e'_2$ \\
Let $(k,e_1,e_2) \in \atomunion{T}$ \\
Then $(k,e'_1,e'_2) \in \atomunion{T}$ 
\end{lemma}
\begin{proof}
The proof is straightforward. Each subgoal follows by induction on the typing derivation of $e_i$ ($i \in \{1,2\})$ .
\end{proof}

\subsubsection{Type substitution preserves subtyping}

\begin{lemma}[Type substitution preserves subtyping]
\label{lm:gobsec-type-subst-preserves-subtyping-n}
\mbox{}
\\
Let $\sigma \in \gsetd{\DeltaX}$ \\
If $\DeltaX,\bigcdot |- S <: S'$ then $\bigcdot;\bigcdot |- \sigma(S) <: \sigma(S')$ \\
If $\DeltaX,\bigcdot |- T <: T'$ then $\bigcdot;\bigcdot |- \sigma(T) <: \sigma(T')$
\end{lemma}

\subsubsection{Type substitution preserves interval subtyping}

\begin{lemma}[Type substitution preserves interval subtyping]
\label{lm:gobsec-type-subst-preserves-interval-subtyping-n}
\mbox{}
\\
Let $\sigma \in \gsetd{\DeltaX}$ and $\DeltaX |- U$ , $\DeltaX |- A$, $\DeltaX |- B$\\
If $\DeltaX |- U \in A .. B$ \\
then $\rhosyn{U} \in \rhosyn{A} .. \rhosyn{B}$
\end{lemma}
\begin{proof}
Apply Lemma~\ref{lm:gobsec-type-subst-preserves-subtyping-n} (\nameref{lm:gobsec-type-subst-preserves-subtyping-n}) for each subgoal.
\end{proof}

\subsubsection{Interval subtyping expansion}
\begin{lemma}[Interval subtyping expansion]
\label{lm:gobsec-interval-subtyping-expansion}
\mbox{}
\\
Let $\sigma \in \gsetd{\DeltaX}$ and $\bigcdot |- U$ , $\DeltaX |- A$, $\DeltaX |- B$\\
If $\bigcdot |- U \in \sigma(A) .. \sigma(B)$ \\
then $\DeltaX |- U \in A .. B$
\end{lemma}

\subsubsection{Downward closed/Monotonicity}
\begin{lemma}[Downward closed/Monotonicity]
\label{lm:gobsec-monotonicity}
\mbox{}
\\
Let $\bigcdot |- S$ \\
If $(k,v_1,v_2) \in \gsetv{S}$ and  $j \le k$  \\
then $(j,v_1,v_2) \in \gsetv{S}$
\end{lemma}
\begin{proof}
The proof is by induction on $S$. All valid cases boil down to  $\stype{\primt}{\primt}$ and $\stype{T}{O}$

\begin{case}[S = $\stype{\primt}{\primt}$]
This is direct from the definition of $\nsetv{\stype{\primt}{\primt}}$
\end{case}
\begin{case}[S = $\stype{T}{O}$]
~

Proof obligations:
\begin{enumerate}
	\item $(j,v_1,v_2) \in \atomunion{T}$. This follows directly from $(k,v_1,v_2) \in \gsetv{S}$
	\item Assuming arbitrary $m,j',T',v'_1,v'_2$ such as:
		\begin{itemize}
			\item $m \in O~\wedge~ \tlookup{\bigcdot;O}{m}{\gmtype{X:A..B}{S_1}{S_2}}$
			\item $j'< j$
			\item $|- T' ~\wedge~ T' \in A..B$
			\item $(j',v_1,v_2) \in \nsetv{S}$
			\item $(j',v'_1,v'_2) \in \nsetv{\ssubst{S_1}{X}{T'}}$
		\end{itemize}
		
		\textbf{Show:}
		\begin{center}
			$(j, \gminv{v_1}{m}{T'}{v'_1}, \gminv{v_2}{m}{T'}{v'_2}) \in \nsetc{\ssubst{S''}{X}{T'}}$
		\end{center}
		
		Instantiate the first conjunct of $(k,v_1,v_2) \in \nsetv{S}$  with $m,j',T',v'_1,v'_2$. Note that:
		\begin{itemize}			
			\item $m \in O~\wedge~ \tlookup{\bigcdot;O}{m}{\gmtype{X:A..B}{S_1}{S_2}}$			
			\item $j'< k$. It follows from $j'<j \le k$
			\item $|- T' ~\wedge~ T' \in A..B$
			\item $(j',v_1,v_2) \in \nsetv{S}$
			\item $(j',v'_1,v'_2) \in \nsetv{\ssubst{S_1}{X}{T'}}$
		\end{itemize}
		
		Hence $(j, \gminv{v_1}{m}{T'}{v'_1}, \gminv{v_2}{m}{T'}{v'_2}) \in \nsetc{\ssubst{S''}{X}{T'}}$
	\item Assuming arbitrary $m,j,v'_1,v'_2$ such as:
		\begin{itemize}
			\item $m \in O \quad \tlookup{\bigcdot,O}{m}{\stype{\primt_1}{\ilab} -> \stype{\primt_2}{\ilab}}$
			\item $(j,v'_1,v'_2) \in \nsetv{\stype{\primt_1}{U_1}}$
		\end{itemize}
		\textbf{Show:}
		\begin{center}
				$(j',v_1.m(v'_1),v_2.m(v'_2)) \in \nsetc{\stype{\primt_2}{\rpolicy(\stype{\primt_1}{U_1},\primt_2)}}$
		\end{center}
		
		Instantiate the second conjunct of $(k,v_1,v_2) \in \nsetv{S}$  with $m,j',v'_1,v'_2$. Note that:
		\begin{itemize}
			\item $m \in O \quad \tlookup{\bigcdot,O}{m}{\stype{\primt_1}{\ilab} -> \stype{\primt_2}{\ilab}}$
			\item $j'< k$. It follows from $j'<j \le k$
			\item $(j,v'_1,v'_2) \in \nsetv{\stype{\primt_1}{U_1}}$
		\end{itemize}
		
		Hence $(j',v_1.m(v'_1),v_2.m(v'_2)) \in \nsetc{\stype{\primt_2}{\rpolicy(\stype{\primt_1}{U_1},\primt_2)}}$
\end{enumerate}
\end{case}

\end{proof}

\subsubsection{Syntactic equivalences implies semantic equivalence}
\syntacticImpliesSemantic*
\begin{proof}

Assuming arbitrary $m,j,v'_1,v'_2$ such as:
\begin{itemize}
	\item $m \in O \quad \tlookup{\bigcdot,O}{m}{\stype{\primt_1}{\ilab} -> \stype{\primt_2}{\ilab}}$
	\item $(j,v'_1,v'_2) \in \nsetv{\stype{\primt_1}{U_1}}$
\end{itemize}

\textbf{Show:}
\begin{center}
		$(j,v_1.m(v'_1),v_2.m(v'_2)) \in \nsetc{\stype{\primt_2}{\rpolicy(\stype{\primt_1}{U_1},\primt_2)}}$
\end{center}

\begin{case}[$U_1 = \primt_1$]
The $v'_1 = v'_2 = \primb'$ and we have \textbf{to show}
	\begin{center}
		$(j,\primb.m(\primb'),\primb.m(\primb')) \in \nsetc{\stype{\primt_2}{\primt_2}}$ \\
		$\equiv (j,\theta(m,\primb,\primb'),\theta(m,\primb,\primb')) \in \nsetc{\stype{\primt_2}{\primt_2}}$
	\end{center}
	
	which follows from the assumption that $\theta$ is partial function that respects that signatures of $\methods(\primt)$.
\end{case}
\begin{case}[$U_1 \neq \primt_1$]
Then we have \textbf{to show}
	\begin{center}
		$(j,\primb.m(v'_1),\primb.m(v'_2)) \in \nsetc{\stype{\primt_2}{\top}}$
	\end{center}
	
	which trivially follows by using Lemma~\ref{lm:gobsec-well-typed-term-related-top} (\nameref{lm:gobsec-well-typed-term-related-top})
\end{case}

\end{proof}

\subsubsection{PER Subtyping}
\begin{lemma}[PER Subtyping]
\label{lm:gobsec-per-types-subtyping-n}
\mbox{}
\\
Let $\bigcdot |- S$ , $\bigcdot |- S'$ and  $\subtps{\bigcdot;\bigcdot}{S'}{S}$\\
(1) If $(k,v_1, v_2) \in \gsetv{S'}$ then $(k,v_1,v_2) \in \gsetv{S}$ \\
(2) If $(k,e_1,e_2) \in \nsetc{S'}$ then $(k,e_1,e_2) \in \nsetc{S}$
\end{lemma}
\begin{proof}

We proof the statements (1) and (2) simultaneously.

Induction on $k$ and then nested induction on $S$. 

We focus on $k > 0$ (the case $k=0$ is trivial).

Statement (1):

All valid cases for $S$ boils down to $\stype{\primt}{\primt}$ and $\stype{T}{O}$.

\begin{case}[S = $\stype{\primt}{\primt}$]
~

Proof obligations:
\begin{enumerate}
	\item $(k,v_1,v_2) \in \atomunion{\primt}$. Apply Lemma~\ref{lm:gobsec-atom-subtyping} (\nameref{lm:gobsec-atom-subtyping})
	\item $v_1 = v_2 = b$. From the third hypothesis we know that $S' = \stype{\primt}{\primt}$ and for the second one we know that $v_1 = v_2 = b$
\end{enumerate}
\end{case}
\begin{case}[S = $\stype{T}{O}$]
~

Denote $S' = \stype{T'}{O'}$

Proof obligations:
\begin{enumerate}
	\item $(k,v_1,v_2) \in \atomunion{T}$. Apply Lemma~\ref{lm:gobsec-atom-subtyping} (\nameref{lm:gobsec-atom-subtyping})
	\item Assuming arbitrary $m,j<k,T',v'_1,v'_2$ such as:
		\begin{itemize}
			\item $m \in O \quad \tlookup{\bigcdot,O}{m}{\gmtype{X:A..B}{S_1}{S_2}}$
			\item $|- T'~\wedge~ T' \in A..B$
			\item $(j,v_1,v_2) \in \nsetv{S}$
			\item $(j,v'_1,v'_2) \in \nsetv{\ssubst{S_1}{X}{T'}}$
		\end{itemize}
		
		\textbf{Show:}
		\begin{center}
		 $(j,\gminv{o_1}{m}{T'}{v'_1}, \gminv{o_2}{m}{T'}{v'_2}) \in \nsetc{\ssubst{S_2}{X}{T'}})$
		\end{center}
		
		Instantiate $(k,v_1, v_2) \in \gsetv{\stype{T'}{O'}}$ with $m,j,T',v'_1,v'_2$. Note that:
		\begin{itemize}
			\item $m \in O' \quad \tlookup{\bigcdot,O'}{m}{\gmtype{X:L'..U'}{S'_1}{S'_2}}$
			\item $j<k$
			\item $|- T'~\wedge~ T' \in A'..B'$.
			\item $(j,v'_1,v'_2) \in \nsetv{\ssubst{S'_1}{X}{T'}}$. Apply the IH with $(j,v'_1,v'_2) \in \nsetv{\ssubst{S_1}{X}{T'}}$ and
				$\subtps{\bigcdot;\bigcdot}{\ssubst{S_1}{X}{T'}}{\ssubst{S'_1}{X}{T'}}$
		\end{itemize}
		
		Hence, $(j,\gminv{o_1}{m}{T'}{v'_1}, \gminv{o_2}{m}{T'}{v'_2}) \in \nsetc{\ssubst{S''_2}{X}{T'}})$
		
		Apply IH, statement (2) with $\subtps{\bigcdot;\bigcdot}{\ssubst{S''_2}{X}{T'}}{\ssubst{S_2}{X}{T'}}$ to obtain
		
		$(j,\gminv{o_1}{m}{T'}{v'_1}, \gminv{o_2}{m}{T'}{v'_2}) \in \nsetc{\ssubst{S_2}{X}{T'}})$
		
	\item Assuming arbitrary $m,j,v'_1,v'_2$ such as:
		\begin{itemize}
			\item $m \in O \quad \tlookup{\bigcdot,O}{m}{\stype{\primt_1}{\ilab} -> \stype{\primt_2}{\ilab}}$
			\item $(j,v'_1,v'_2) \in \nsetv{\stype{\primt_1}{U_1}}$
		\end{itemize}
		\textbf{Show:}
		\begin{center}
				$(j,v_1.m(v'_1),v_2.m(v'_2)) \in \nsetc{\stype{\primt_2}{\rpolicy(\stype{\primt_1}{U_1},\primt_2)}}$
		\end{center}
		
		Do a case analysis on $S'$. Each case reduces to $S'= \stype{T'}{O'}$ or $S'= \stype{\primt'}{\primt'}$
		\begin{case}[$S'= \stype{T'}{O'}$]
			Instantiate $(k,v_1, v_2) \in \gsetv{\stype{T'}{O'}}$ with $m,j,v'_1,v'_2$. Note that:
			\begin{itemize}
				\item $m \in O \quad \tlookup{\bigcdot,O}{m}{\stype{\primt_1}{\ilab} -> \stype{\primt_2}{\ilab}}$. Recall, that there 
				is no subtyping rules between primitive types.
				\item $(j,v'_1,v'_2) \in \nsetv{\stype{\primt_1}{U_1}}$
			\end{itemize}	
			
			Hence $(j,v_1.m(v'_1),v_2.m(v'_2)) \in \nsetc{\stype{\primt_2}{\rpolicy(\stype{\primt_1}{U_1},\primt_2)}}$
		
		\end{case}
		\begin{case}[$S'= \stype{\primt'}{\primt'}$]
		~
		
		It means that $v_1 = v_2 = b$. Apply Lemma~\ref{lm:syntactic-implies-semantic} (\nameref{lm:syntactic-implies-semantic}) with 
		$b$,$\primt'$ and $O$. Note that $\stypeof{~}{b}{\primt'}$ and $\primt' <: O$.

		\end{case}
\end{enumerate}
\end{case}

Statement (2):

Denote $S \defas \stype{T}{U}$ and $S' \defas \stype{T'}{U'}$

Proof obligations:
\begin{enumerate}
	\item $(k,e_1,e_2) \in \atomunion{T}$. Apply Lemma \ref{lm:gobsec-atom-subtyping} (\nameref{lm:gobsec-atom-subtyping}) 
	with $(k,e_1,e_2) \in \atomunion{T'}$ ($(k,e_1,e_2) \in \nsetc{S'}$) and
			$\subtps{\bigcdot;\bigcdot}{T'}{T}$
	\item Assuming arbitrary $j<k,v_1,v_2$ such as $j<k$: 
		\begin{itemize}
			\item $e_1 \reduce^{j} v_1$
			\item $e_2 \reduce^{j} v_2$
		\end{itemize}
		
		\textbf{Show:}		
		\begin{center}
			$(j,v_1,v_2) \in \gsetv{S}$
		\end{center}
		
		Instantiate $(k,e_1,e_2) \in \nsetc{S'}$  with $j,v_1,v_2$ to obtain:
		$(k,v_1,v_2) \in \nsetv{S'}$.
		
		Apply the IH, statement (1) with $(k,v_1,v_2) \in \nsetv{S'}$ and
		$\subtps{\bigcdot;\bigcdot}{S'}{S}$ to obtain:
		
		$(k,v_1,v_2) \in \nsetv{S}$	
\end{enumerate}
\end{proof}

\subsubsection{Anti reduction}
\begin{lemma}[Anti reduction]
\label{lm:gobsec-anti-reduction}
\mbox{}
\\
Let $S \defas \stype{T}{U}$ \\
Let $(j,e_1,e_2) \in \atomunion{T}$ \\
Let $j' \le j$ and $j \leq j' + k$ \\
Let $e_1 \reduce^{\le k} e'_1$ and $e_2 \reduce^{\le k} e'_2$ \\
Let $(j',e'_1,e'_2) \in \nsetc{S}$ \\
Then $(j,e_1,e_2) \in \nsetc{S}$
\end{lemma}
\begin{proof}
Denote $S \defas \stype{T}{U}$.

Proof obligations:
\begin{enumerate}
	\item $(j,e_1,e_2) \in \atomunion{T}$. Apply Lemma~\ref{lm:gobsec-atom-reduction} (\nameref{lm:gobsec-atom-reduction}) with 
	$(j',e'_1,e'_2) \in \atomunion{T}$ which follows from $(j',e'_1,e'_2) \in \nsetc{S}$
	\item Assuming $j_1<j,v_1,v_2$ such as:
		\begin{itemize}
			\item $e_1 \reduce^{\le j_1} v_1$
			\item $e_1 \reduce^{\le j_1} v_2$
		\end{itemize}
		
		\textbf{Show}:
		\begin{center}
			$(j-j_1,v_1,v_2) \in \nsetv{S}$
		\end{center}
		
		We have that:
		\begin{center}
			$e_1 \reduce^{\le k} e'_1 \reduce^{j'_1} v_1$ \\
			$e_2 \reduce^{\le k} e'_2 \reduce^{j'_1} v_2$
		\end{center}
		
		where $j'_1 < j'$.

		Instantiate $(j',e'_1,e'_2) \in \nsetc{S}$ with $j'_1$. Note that $j'_1 < j'$.
		
		Hence, $(j'- j'_1,v_1,v_2) \in \nsetv{S}$.
		
		Apply Lemma~\ref{lm:gobsec-monotonicity} (\nameref{lm:gobsec-monotonicity}) with $j - j_1 \le j - (k + j'_1) \le j' - j'_1$ ($ j-k \le j'$) to obtain: 
		
		$(j - j_1,v_1,v_2) \in \nsetv{S}$
\end{enumerate}
\end{proof}

\subsubsection{Monadic bind}
\begin{lemma}[Monadic bind]
\label{lm:gobsec-monadic-bind-n}
\mbox{}
\\
If $(k,e_1,e_2) \in \gsetc{S}$ \\
and $\forall j \le k. \forall v_1,v_2.~(j,v_1,v_2) \in \gsetv{S} \implies (j,\inhole{E}{v_1},\inhole{E}{v_2}) \in \gsetc{S'}$ \\
then $(k,\inhole{E}{e_1},\inhole{E}{e_2}) \in \gsetc{S'}$
\end{lemma}
\begin{proof}
Let us assume that $e_2 \reduce^{\le j'} v'_1$ and $e_2 \reduce^{\le j'} v'_2$ where $j' \leq k$ (in other case the lemma vacuously holds)

Instantiate $(k,e_1,e_2) \in \gsetc{S}$ with $j',v'_1,v'_2$. 

Hence, $(k-j',v'_1,v'_2) \in \nsetv{S'}$

By the dynamic semantics we have know that:
\begin{center}
	$\inhole{E}{e_1} \reduce^{\le j'} \inhole{E}{v'_1}$ \\
	$\inhole{E}{e_2} \reduce^{\le j'} \inhole{E}{v'_2}$
\end{center}

Instantiate the second premise with $k-j',v'_1,v'_2$. Note that $k-j' \le k$ and $(k-j',v'_1,v'_2) \in \nsetv{S}$.

Hence, $(k-j',\inhole{E}{v'_1},\inhole{E}{v'_2}) \in \gsetc{S'}$

Instantiate Lemma~\ref{lm:gobsec-anti-reduction} (\nameref{lm:gobsec-anti-reduction}). Note that:
\begin{itemize}
	\item $k - j' \le k$
	\item $ k \leq k-j' + j'$
	\item	$\inhole{E}{e_1} \reduce^{\le j'} \inhole{E}{v'_1}$ 
	\item $\inhole{E}{e_2} \reduce^{\le j'} \inhole{E}{v'_2}$
	\item $(k-j',\inhole{E}{v'_1},\inhole{E}{v'_2}) \in \gsetc{S'}$
\end{itemize}

Hence, $(k,\inhole{E}{e_1},\inhole{E}{e_2}) \in \gsetc{S'}$
\end{proof}

\subsubsection{Substitutions preserve simple typing}
\begin{lemma}[Substitutions preserve simple typing]
\label{lm:gobsec-substitution-in-atoms}
\mbox{}
\\
Let $\DeltaX,\Gamma |- e_1 : \stype{T_1}{U_1}$ and  $\DeltaX,\Gamma |- e_2 : \stype{T_1}{U_1}$ \\
Let $\sigma \in \nsetd{\DeltaX}$ and $(k,\gamma_1,\gamma_2) \in \nsetg{\sigma(\Gamma)}$ \\
Then $(k,\sigma(\gamma_1(e_1)),\sigma(\gamma_2(e_2))) \in \atomunion{\sigma(T_1)}$
\end{lemma}
\begin{proof}
The proof is straightforward. Then the goal is equivalent to show:
$(k,\gamma_1(e_1),\gamma_2(e_2)) \in \atomunion{\sigma(T_1)}$ (because type variable are not taken into account by the simple type system). 

For each value substitution apply Lemma~\ref{lm:gobsec-value-substitution-preserve-simple-typing} (\nameref{lm:gobsec-value-substitution-preserve-simple-typing}) with
$\Gamma |-_1 : e_i$ and $\gamma_i \models \Gamma$ to
obtain $|-_{1} \gamma_i(e_i) : \sigma(T_1)$
\end{proof}

\subsubsection{Well-typed terms are related at top}
\begin{lemma}[Well-typed terms are related at top]
\label{lm:gobsec-well-typed-term-related-top}
\mbox{}
\\
Let $\sigma \in \nsetd{\DeltaX}$ and $(k,\gamma_1,\gamma_2) \in \nsetg{\sigma(\Gamma)}$ \\
Let $\DeltaX;\Gamma |- e_1: \stype{T}{\top}$ \\
Let $\DeltaX;\Gamma |- e_2: \stype{T}{\top}$ \\
Then $(k,\sigma(\gamma_1(e_1)),\sigma(\gamma_2(e_2))) \in \nsetc{\stype{\sigma(T)}{\top}}$
\end{lemma}
\begin{proof}
Proof obligations 
\begin{enumerate}
	\item $(k,\sigma(\gamma_1(e_1)), \sigma(\gamma_1(e_2))) \in \atomunion{\sigma(U)}$. It follows from
	Lemma~\ref{lm:gobsec-substitution-in-atoms}~(\nameref{lm:gobsec-substitution-in-atoms}).
	\item Assuming $j<k,v_1, v_2$ such as: 
		\begin{itemize}
			\item $\sigma(\gamma_1(e_1)) \reduce^{\le j} v_1$
			\item $\sigma(\gamma_2(e_2)) \reduce^{\le j} v_2$
		\end{itemize}
		
		\textbf{Show}
		\begin{center}
			$(k-j,v_1,v_2) \in \gsetv{\stype{\sigma(T)}{\top}}$
		\end{center}
		
		Which is equivalent to show
		\begin{center}
			$(k-j,v_1,v_2) \in \atomunion{\sigma(T)}$
		\end{center}
		Apply Lemma ~\ref{lm:gobsec-atom-reduction}~(\nameref{lm:gobsec-atom-reduction}) with  
		$(k,\sigma(\gamma_1(e_1)), \sigma(\gamma_1(e_2))) \in \atomunion{\sigma(T)}$, $v_1$ and $v_2$
		to obtain
		
		$(k-j,v_1,v_2) \in \atomunion{\sigma(T)}$
\end{enumerate}
\end{proof}

\subsubsection{Related values are related terms}
\begin{lemma}[Related values are related terms]
\label{lm:gobsec-related-values-related-terms-n}
If $(k,v_1,v_2) \in \gsetv{S}$ then $(k,v_1,v_2) \in \gsetc{S}$
\end{lemma}
\begin{proof}
The proof trivially follows.
\end{proof}

\clearpage
\subsection{Proof of Fundamental Property}

\subsubsection{Pre-compatibility: Method Invocation}
\begin{lemma}[Pre-compatibility: Method Invocation]
\label{lm:gobsec-pre-compatibility-method-inv-n}
\mbox{}
\\
Let $\sigma \in \gsetd{\DeltaX},~\wft{\DeltaX}{\stype{T_1}{U_1}}$ \\
Let $k'' \le k' \le k$ \\
Let $(k',v_1,v_2) \in \gsetv{\sigma(\stype{T_1}{U_1})}$ \\
Let $\DeltaX |- m \in U_1$ and $\tlookup{\DeltaX;U_1}{m}{\gmtype{X:A..B}{S_2}{S}}$\\
Let $\DeltaX |- U' ~\wedge~ \DeltaX |- U' \in A..B$\\
Let $(k,v''_1,v''_2) \in \gsetv{\sigma(\ssubst{S_2}{U'}{X})}$\\
then $(k'',\gminv{v_1}{m}{\rhosyn{U'}}{v''_1},\gminv{v_2}{m}{\rhosyn{U'}}{v''_2}) \in \gsetc{\sigma(\ssubst{S}{U'}{X})}$
\end{lemma}
\begin{proof}	
	Instantiate $(k',v_1,v_2) \in \gsetv{\sigma(\stype{T_1}{U_1})}$ with:	
	$m, k'',\rhosyn{U'},v''_1,v''_2$. Note that:
	\begin{itemize}
		\item $m \in \sigma(T_1)$. It follows from $\DeltaX |- m \in U_1$, $\sigma \in \nsetd{\DeltaX}$ and $\wft{\DeltaX}{\stype{T_1}{U_1}}$. Then 
		 
		$\tlookup{\bigcdot,\sigma(T_1)}{m}{\gmtype{X:\sigma(A)..\sigma(B)}{\sigma(S_2)}{\sigma(S)}}$.
		\item $k'' < k$ which follows directly from hypothesis.
		\item $|- \rhosyn{U'}$ which is direct from $\DeltaX |- U'$ and $\rho \in \gsetd{\DeltaX}$.
		\item $\rhosyn{U'} \in \sigma(A) .. \sigma(B)$. Apply Lemma~\ref{lm:gobsec-type-subst-preserves-interval-subtyping-n}
		(\nameref{lm:gobsec-type-subst-preserves-interval-subtyping-n}) with $\sigma \in \gsetd{\DeltaX}$ and $\DeltaX |- U' \in A..B$
		\item $(k'',v''_1,v''_2) \in \gsetv{\ssubst{\sigma(S_2)}{\sigma(U')}{X}}$. It follows from
			\begin{itemize}
				\item $(k,v''_1,v''_2) \in \gsetv{\ssubst{\sigma(S_2)}{\sigma(U')}{X}}$. Note that $\sigma(\ssubst{S_2}{U'}{X}) = \ssubst{\sigma(S_2)}{\sigma(U')}{X}$
				and $(k,v''_1,v''_2) \in \gsetv{\sigma(\ssubst{S_2}{U'}{X})}$ is given in hypothesis
				\item Apply Lemma~\ref{lm:gobsec-monotonicity} (\nameref{lm:gobsec-monotonicity}) with $(k,v''_1,v''_2) \in \gsetv{\ssubst{\sigma(S_2)}{\sigma(U')}{X}}$  and $k'' \le k$ we obtain 
				$(k'',v''_1,v''_2) \in \gsetv{\ssubst{\sigma(S_2)}{\sigma(U')}{X}}$.
			\end{itemize}
	\end{itemize}
	
	Hence, $(k'',\gminv{v_1}{m}{\rhosyn{U'}}{v''_1},\gminv{v_2}{m}{\rhosyn{U'}}{v''_2}) \in \gsetc{\ssubst{\sigma(S)}{\sigma(U')}{X}}$
	
	Note $\sigma(\ssubst{S}{U'}{X}) = \ssubst{\sigma(S)}{\sigma(U')}{X}$.
	
	Hence, $(k'',\gminv{v_1}{m}{\rhosyn{U'}}{v''_1},\gminv{v_2}{m}{\rhosyn{U'}}{v''_2}) \in \gsetc{\sigma(\ssubst{S}{U'}{X})}$
\end{proof}

\subsubsection{Compatibility-Var}
\begin{lemma}[\gobsec Compatibility-Var]
\label{lm:gobsec-compatibility-var-n}
\mbox{}

$\DeltaX;\Gamma |- x \approx x: \Gamma(x)$
\end{lemma}
\begin{proof}
  First, let us denote $S \defas \Gamma(x)$.
	
	Proof obligations:
	\begin{enumerate}
		\item $\DeltaX;\Gamma |- x: S$ which is direct.
		\item Assuming arbitrary $k,\sigma,\gamma_1,\gamma_2$ such as:			
			$k \ge 0, \sigma \in \nsetd{\DeltaX},\gamma_1,\gamma_2,~(k,\gamma_1, \gamma_2) \in \nsetg{\rho(\Gamma)}$
			
			\textbf{Show}
			\begin{center}
			$(k, \rhosyn{\gamma_1(x)}, \rhosyn{\gamma_2(x)}) \in \gsetc{\sigma(S)}$ \\
			$\equiv (k,\gamma_1(x),\gamma_2(x)) \in \gsetc{\sigma(S)}$
			\end{center}
			
			From $(k,\gamma_1,\gamma_2) \in \gsetg{\sigma(\Gamma)}$ we know that exists $v_1,v_2$ such as: 
			\begin{itemize}
				\item $\gamma_1(x) = v_1$
				\item $\gamma_2(x) = v_2$
				\item $(k,v_1,v_2) \in \gsetv{\sigma(S)}$
			\end{itemize}
			
		 Apply Lemma~\ref{lm:gobsec-related-values-related-terms-n} (\nameref{lm:gobsec-related-values-related-terms-n}) 
			with $(k,\gamma_1(x),\gamma_2(x)) \in \gsetv{\sigma(S)}$ to obtain ${(k,\gamma_1(x),\gamma_2(x)) \in \gsetc{\sigma(S)}}$
		\end{enumerate}			
\end{proof}

\subsubsection{Compatibility-Prim}

\begin{lemma}[\gobsec Compatibility-Prim]
\label{lm:gobsec-compatibility-prim-n}
\mbox{}

Let $\primt = \Delta_{\primb}(\primb)$.

Then $\DeltaX;\Gamma |- b \approx b: \stype{\primt}{\primt}$
\end{lemma}
\begin{proof}
Proof obligations:
\begin{enumerate}
	\item $\DeltaX;\Gamma |- b : \stype{\primt}{\primt}$. Apply rule (TPrim)
	\item Assuming arbitrary $k,\rho, \gamma_1,\gamma_2$ such as:
		\begin{itemize}
			\item $k \ge 0, \sigma \in \nsetd{\DeltaX},~(k,\gamma_1,\gamma_2) \in \nsetg{\sigma(\Gamma)}$
		\end{itemize}
		\textbf{Show:}
		\begin{center}
			$(k, \rhosyn{\gamma_1(b)}, \rhosyn{\gamma_2(b)}) \in \nsetc{\sigma(\stype{\primt}{\primt})}$ \\
			$\equiv (k, b, b) \in \nsetc{\stype{\primt}{\primt}}$ 
		\end{center}
		
		Apply Lemma~\ref{lm:gobsec-related-values-related-terms-n}~(\nameref{lm:gobsec-related-values-related-terms-n})
		with $(k,b,b) \in \nsetv{\stype{\primt}{\primt}}$ to obtain:
		
		$(k, b, b) \in \nsetc{\stype{\primt}{\primt}}$ 
\end{enumerate}
\end{proof}

\subsubsection{Compatibility Subsumption}
\begin{lemma}[\gobsec Compatibility-Subsumption]
\label{lm:gobsec-compatibility-sub-n}
\mbox{}

Let $\DeltaX;\Gamma |- e_1  \approx e_2 : S'$. Let $\DeltaX; \bigcdot |-S' <: S$.

Then $\DeltaX;\Gamma |- e_1  \approx e_2 : S$.
\end{lemma}
\begin{proof}
Proof obligations:
\begin{enumerate}
	\item $\DeltaX;\Gamma |-  e_1 : S$ and  $\DeltaX;\Gamma |-  e_2 : S$. Apply rule (TSub) with 
	$\DeltaX;\Gamma |- e_i : S'$ (obtained from $\DeltaX;\Gamma |- e_1  \approx e_2 : S'$) and $\DeltaX; \bigcdot |-S' <: S$
	\item Assuming arbitrary $k,\sigma,\gamma_1,\gamma_2$ such as : 
		\begin{itemize}
			\item $k \ge 0, \rho \in \nsetd{\DeltaX},~(k,\gamma_1,\gamma_2) \in \nsetg{\sigma(\Gamma)}$
		\end{itemize}
		
		\textbf{Show:} 
		\begin{center}
		$(k,\rhosyn{\gamma_1(e_1)},\rhosyn{\gamma_2(e_2)}) \in \nsetc{\rho(S)}$.
		\end{center}
		
		Instantiate  $\DeltaX;\Gamma |- e_1  \approx e_2 : S'$ with $k,\sigma,\gamma_1,\gamma_2$ to obtain:
		
		$(k,\rhosyn{\gamma_1(e_1)}),\rhosyn{\gamma_2(e_2)}) \in \gsetc{\sigma(S')}$. 
		
		Apply Lemma~\ref{lm:gobsec-per-types-subtyping-n} (\nameref{lm:gobsec-per-types-subtyping-n}) with 
		$\bigcdot; \bigcdot |- \sigma(S') <: \sigma(S)$. Note that:
		\begin{itemize}
			\item $\bigcdot; \bigcdot |- \sigma(S') <: \sigma(S)$ follows from Lemma \ref{lm:gobsec-type-subst-preserves-subtyping-n}
			~(\nameref{lm:gobsec-type-subst-preserves-subtyping-n}) applied to $\sigma \in \nsetd{\DeltaX}$ and 
			$\DeltaX; \bigcdot |-S' <: S$.
		\end{itemize}
		
		Hence, $(k,\rhosyn{\gamma_1(e_1)},\rhosyn{\gamma_2(e_2)}) \in \nsetc{\rho(S)}$.
\end{enumerate}
\end{proof}

\subsubsection{Compatibility Object}
\begin{lemma}[\gobsec Compatibility-Object]
\label{lm:gobsec-compatibility-object-n}
\mbox{}
\\
Let be $S \triangleq \stype{O}{U} \quad$ \\
Let be $O \triangleq \rtypex{\alpha}{\gntrecordx{m}{\gmtype{X : A..B}{S'}{S''}}}$ \\
Then: 
\begin{mathpar}	
  \inference{
			\tlookup{\DeltaX,O}{m_i}{\gmtype{X:U_{li}..U_{ui}}{S'_i}{S''_i}}	\\
			\DeltaX,X:U_{li}..U_{ui};\Gamma, z: S, x:S^{'}_i |- e_i \approx e'_i: {S''_i}
    }{
      \DeltaX;\Gamma |- \objectxx{z}{S}{\overline{m(x)e}} \approx \objectxx{z}{S}{\overline{m(x)e'}}: S
    }
\end{mathpar}

\begin{proof}
Denote $o = \objectxx{z}{S}{\overline{m(x)e}}$ and $o' = \objectxx{z}{S}{\overline{m(x)e'}}$

Proof obligations:
\begin{enumerate}
	\item $\DeltaX;\Gamma |- o : S ~\wedge~ \DeltaX;\Gamma |-  o': S$. Apply rule (TObj)
	\item Consider arbitrary $k,\sigma, \gamma_1,\gamma_2$ such as: $k \geq 0, \sigma \in \gsetd{\DeltaX}, 
		(k,\gamma_1,\gamma_2) \in \gsetg{\sigma(\Gamma)}$
		
		\textbf{Show:} 
		\begin{center}
		$ (k,\rhosyn{\gamma_1(o)},\rhosyn{\gamma_2(o')}) : \gsetc{\sigma(S)}$ \\
		$\equiv  (k,  \objectxx{z}{\rhosyn{S}}{\overline{m(x) \rhosyn{\gamma_1(e)}}},\objectxx{z}{\rhosyn{S}}{\overline{m(x)\rhosyn{\gamma_2(e')}}}) : \gsetc{\sigma(S)}$
		\end{center}
		
		Apply Lemma~\ref{lm:gobsec-related-values-related-terms-n} (\nameref{lm:gobsec-related-values-related-terms-n}) to transform the goal to 
		\begin{center}
		$(k,  \objectxx{z}{\rhosyn{S}}{\overline{m(x) \rhosyn{\gamma_1(e)}}},\objectxx{z}{\rhosyn{S}}{\overline{m(x)\rhosyn{\gamma_2(e')}}}) : \gsetv{\sigma(S)}$
		\end{center}
		
		Let us denote $o_1 = \objectxx{z}{\rhosyn{S}}{\overline{m(x) \rhosyn{\gamma_1(e)}}}, o_2 = \objectxx{z}{\rhosyn{S}}{\overline{m(x)\rhosyn{\gamma_2(e')}}}$

		By well-formedness of the type $S$ (Figure~\ref{fig:gbosec-ok-types-subtyping}) we know that $\rho(U)$ is necesarily an object type (\ie it is not a primitive type)
		
		Proof of $(k,o_1,o_2) \in \gsetv{\stype{\sigma(O)}{\sigma(U)}}$.
		
		Sub goals 
		\begin{itemize}
			\item $(k,o_1,o_2) \in \atomunion{\sigma(O)}$. Apply Lemma~\ref{lm:gobsec-substitution-in-atoms} (\nameref{lm:gobsec-substitution-in-atoms}) with 
			$\DeltaX;\Gamma |- o : S$, $\DeltaX;\Gamma |-  o': S$, $\sigma \in \nsetd{\DeltaX}$ and $(k,\gamma_1,\gamma_2) \in \nsetg{\sigma(\Gamma)}$
			\item Assuming arbitrary $m,j,T',v'_1,v'_2$ such as: 
			\begin{itemize}
				\item $m \in \rho(U) \quad \tlookup{\bigcdot,\sigma(O)}{m}{\gmtype{X:\sigma(A)..\sigma(B)}{\sigma(S')}{\sigma(S'')}}$
				\item $j<k$
				\item $|- T'~\wedge~T' \in \sigma(A) .. \sigma(B)$			
				\item $(j,o_1,o_2) \in \nsetv{\sigma(S)}$
				\item $(j, v_1', v_2') \in \nsetv{\sigma(S')}$
			\end{itemize}
			
		  \textbf{Show}:
			\begin{center}
				$(j, \gminv{o_1}{m}{T'}{v_1'}, \gminv{o_2}{m}{T'}{v_2'}) \in \nsetc{\sigma(S'')}$ 
			\end{center}
			
			Denote $\methimpl{o_1}{m}{x.\rhosyn{\gamma_1(e)}}$ and $\methimpl{o_1}{m}{x.\rhosyn{\gamma_2(e')}}$
			
	    Then, the above goal rewrites to
			\begin{center}
			$(j, 
				\ssubst{\ssubst{\ssubst{\rhosyn{\gamma_1(e)}}{T'}{X}}{o_1}{z}}{v_1'}{x},  
				\ssubst{\ssubst{\ssubst{\rhosyn{\gamma_2(e')}}{T'}{X}}{o_2}{z}}{v_2'}{x}) \in \nsetc{\sigma(S'')}$ \\
			$(j, 
				\ssubst{\ssubst{\xsigma{X}{T'}(\gamma_1(e))}{o_1}{z}}{v_1'}{x},  
				\ssubst{\ssubst{\xsigma{X}{T'}(\gamma_2(e'))}{o_2}{z}}{v_2'}{x}) \in \nsetc{\sigma(S'')}$
			\end{center}
			
			Instantiate the second conjunct of the IH $\DeltaX,X:A..B;\Gamma, z: S, x:S^{'} |- e \approx e': {S''}$ 
		  with 
			
			$j, \sigma' = \xsigma{X}{T'}, \gamma'_1  = \extgammax{\gamma_1}{z}{o_1}{x}{v'_1}, 
			\gamma'_1 =\extgammax{\gamma_2}{z}{o_2}{x}{v'_2}$. Note that:
			\begin{itemize}
				\item $j \geq 0$
				\item $\xsigma{X}{T'} \in \nsetd{\DeltaX,X:A..B}$. It follows from:
					\begin{itemize}
						\item $\sigma \in \nsetd{\DeltaX}$. It follows from above.
						\item $\DeltaX |- T' \in A ..B$. Apply Lemma~\ref{lm:gobsec-interval-subtyping-expansion} 
						(\nameref{lm:gobsec-interval-subtyping-expansion}) with $\sigma \in \nsetd{\DeltaX}$ and $|- T'$ and $T' \in \sigma(A) .. \sigma(B)$
					\end{itemize}
				\item $(j, \gamma'_1,\gamma'_2) \in \nsetg{\sigma(\Gamma), z: \sigma(S), x:\sigma(S')}$. It follows from:
					\begin{itemize}
						\item $(j, \gamma_1,\gamma_2) \in \nsetg{\sigma(\Gamma)}$. It follows from above.
						\item $(j, o_1,o_2) \in \nsetv{\sigma(S)}$. It follows from the above.
						\item $(j, v'_1,v'_2) \in \nsetv{\sigma(S')}$. It follows from the above.
					\end{itemize}
			\end{itemize}
			
			Hence $(j,\xsigma{X}{T'}(\extgammax{\gamma_1}{z}{o_1}{x}{v'_2}(e)),\xsigma{X}{T'}(\extgammax{\gamma_2}{z}{o_2}{x}{v'_2}(e'))) \in \nsetc{\sigma(S'')}$
				
			Since $o_1,o_2,v'_1, v'_2$ are closed values with respect to type variables we can rewrite this as:
		  \begin{center}		
				$(j, 
				\ssubst{\ssubst{\xsigma{X}{T'}(\gamma_1(e))}{o_1}{z}}{v_1'}{x},  
				\ssubst{\ssubst{\xsigma{X}{T'}(\gamma_2(e'))}{o_2}{z}}{v_2'}{x}) \in \nsetc{\sigma(S'')}$
			\end{center}		
		\end{itemize}
\end{enumerate}
\end{proof}
\end{lemma}

\clearpage
\subsubsection{Compatibility Method Invocation Declassification}

\begin{lemma}[\gobsec Compatibility-Method-Invocation-Declassification]
\label{lm:gobsec-compatibility-method-invocation-decl-n}
\mbox{}
\\
Let $S_1 \defas \stype{T_1}{U_1},S_2 \defas \stype{T_2}{U_2}$ \\
Let $\DeltaX;\Gamma |- e_1 \approx e'_1 : \stype{T_1}{U_1}$ \\
Let $\DeltaX |- m \in U_1,~\tlookup{\DeltaX;U_1}{m}{\gmtype{X:A..B}{S_2}{S}}$	\\
Let $\DeltaX |- U' ~\wedge~ \DeltaX |- U' \in A..B$\\
Let $\DeltaX;\Gamma |- e_2 \approx e'_2 : \ssubst{S_2}{U'}{X}$ \\
Then $\DeltaX;\Gamma |- \gminv{e_1}{m}{U'}{e_2} \approx \gminv{e'_1}{m}{U'}{e'_2}: \ssubst{S}{U'}{X}$
\end{lemma}

\begin{proof}
Let us denote $e = \gminv{e_1}{m}{U'}{e_2}$ and $e' = \gminv{e'_1}{m}{U'}{e'_2}$

\noindent
Proof obligations:
\begin{enumerate}
	\item $\DeltaX;\Gamma |- e : \ssubst{S}{U'}{X}$ and $\DeltaX;\Gamma |- e' : \ssubst{S}{U'}{X}$ which follow directly 
	 from the premises and the rule (TmD).
	\item Assuming arbitrary $k, \sigma, \gamma_1, \gamma_2$ such as: 
		$k \geq 0,\sigma \in \nsetd{\DeltaX}, (k,\gamma_1,\gamma_2) \in \nsetg{\sigma(\Gamma)}$
		
		\textbf{Show:}
		\begin{center}
			$(k, \sigma(\gamma_1(\gminv{e_1}{m}{U'}{e_2})), \sigma(\gamma_2(\gminv{e'_1}{m}{U'}{e'_2}))) \in \nsetc{\sigma(\ssubst{S}{U'}{X})}$ \\		
			$\equiv 
			(k, 
			\gminv{\sigma(\gamma_1(e_1))}{m}{\sigma(U')}{\sigma(\gamma_1(e_2))}, 
			\gminv{\sigma(\gamma_2(e'_1))}{m}{\sigma(U')}{\sigma(\gamma_2(e'_2))}) \in \nsetc{\sigma(\ssubst{S}{U'}{X})}$ 		
		\end{center}
		(because $X \notin \dom(\sigma)$)
			
		Instantiate the hypothesis $\DeltaX;\Gamma |- e_1 \approx e'_1 : \stype{T_1}{U_1}$ with $k, \sigma,\gamma_1,\gamma_2$, hence:
		
		$(k, \rhosyn{\gamma_1(e_1)},\rhosyn{\gamma_2(e'_1)}) \in \nsetc{\sigma(\stype{T_1}{U_1})}$
		
		Let $k' \le k$ and let $(k',v_1,v_2) \in \gsetv{\sigma(\stype{T_1}{U_1})}$. By Lemma~\ref{lm:gobsec-monadic-bind-n} (\nameref{lm:gobsec-monadic-bind-n}) we can rewrite the goal \textbf{to show}
		\begin{center}
		$(k', \gminv{v_1}{m}{\rhosyn{U'}}{\rhosyn{\gamma_1(e_2)}}, \gminv{v_2}{m}{\rhosyn{U'}}{\rhosyn{\gamma_2(e'_2)}})) \in \gsetc{\ssubst{S}{U'}{X}}$ 
		\end{center}
		
		Instantiate the hypothesis $\DeltaX;\Gamma |- e_2 \approx e'_2 : \ssubst{S_2}{U'}{X}$ with $k, \rho,\gamma_1,\gamma_2$, hence: 
		
		$(k, \rhosyn{\gamma_1(e_2)},\rhosyn{\gamma_2(e'_2)} \in \gsetc{\sigma(\ssubst{S_2}{U'}{X})}$
		
		Let $k'' \le k'$ and let $(k'',v''_1,v''_2) \in \gsetv{\sigma(\ssubst{S_2}{U'}{X})}$. By Lemma~\ref{lm:gobsec-monadic-bind-n} (\nameref{lm:gobsec-monadic-bind-n}) we can rewrite the goal \textbf{to show}
		\begin{center}
		$(k'', \gminv{v_1}{m}{\rhosyn{U'}}{v''_1}, \gminv{v_2}{m}{\rhosyn{U'}}{v''_2})) \in \gsetc{\sigma(\ssubst{S}{U'}{X})}$ 
		\end{center}		

		Then, we apply the Lemma~\ref{lm:gobsec-pre-compatibility-method-inv-n} (\nameref{lm:gobsec-pre-compatibility-method-inv-n}) 
		with :
		\begin{itemize}
			\item $\sigma \in \gsetd{\DeltaX}$
			\item $\wft{\DeltaX}{\stype{T_1}{U_1}}$
			\item $(k',v_1,v_2) \in \gsetv{\sigma(\stype{T_1}{U_1})}$ 
			\item $\DeltaX |- m \in U_1$ and $\tlookup{\DeltaX;U_1}{m}{\gmtype{X:A..B}{S_2}{S}}$
			\item $\DeltaX |- U' ~\wedge~ \DeltaX |- U' \in A..B$
			\item $(k'',v''_1,v''_2) \in \gsetv{\sigma(\ssubst{S_2}{U'}{X})}$
		\end{itemize}								
		
		to obtain: $(k'', \gminv{v_1}{m}{\rhosyn{U'}}{v''_1}, \gminv{v_2}{m}{\rhosyn{U'}}{v''_2})) \in \gsetc{\sigma(\ssubst{S}{U'}{X})}$ 		
				
\end{enumerate}

\end{proof}

\subsubsection{Compatibility-Method-Invocation-High}
\begin{lemma}[\gobsec Compatibility-Method-Invocation-High]
\label{lm:gobsec-compatibility-method-invocation-high-n}
\mbox{}
\\
Let $S_1 \defas \stype{T_1}{U_1},S_2 \defas \stype{T_2}{U_2}$ \\
If $\DeltaX;\Gamma |- e_1 \approx e'_1 : \stype{T_1}{U_1}$ \\
$\DeltaX |- m \notin U_1,~\tlookup{\DeltaX;T_1}{m}{\gmtype{X:A..B}{S_2}{\stype{T}{U}}}$	\\
$\DeltaX |- U' ~\wedge~ \DeltaX |- U' \in A..B$\\
$\DeltaX;\Gamma |- e_2 \approx e'_2 : \ssubst{S_2}{U'}{X}$ \\
then $\DeltaX;\Gamma |- \gminv{e_1}{m}{U'}{e_2} \approx \gminv{e'_1}{m}{U'}{e'_2}: \stype{\ssubst{T}{U'}{X}}{\top}$
\end{lemma}
\begin{proof}

Proof obligations:
\begin{enumerate}
	\item $\DeltaX;\Gamma |- \gminv{e_1}{m}{U'}{e_2}: \stype{\ssubst{T}{U'}{X}}{\top}$ and $\DeltaX;\Gamma |- \gminv{e'_1}{m}{U'}{e'_2}: \stype{\ssubst{T}{U'}{X}}{\top}$. Apply rule (TmH).
	\item  Assuming arbitrary $k, \sigma, \gamma_1, \gamma_2$ such as: 
		$k \geq 0,\sigma \in \nsetd{\DeltaX}, (k,\gamma_1,\gamma_2) \in \nsetg{\sigma(\Gamma)}$
		
		\textbf{Show:}
		\begin{center}
			$(k, \sigma(\gamma_1(\gminv{e_1}{m}{U'}{e_2})), \sigma(\gamma_2(\gminv{e'_1}{m}{U'}{e'_2}))) \in \gsetc{\stype{\sigma(\ssubst{T}{U'}{X})}{\top}}$ 
		\end{center}
		
		Apply Lemma~\ref{lm:gobsec-well-typed-term-related-top} (\nameref{lm:gobsec-well-typed-term-related-top}) with:
		\begin{itemize}
			\item $\DeltaX;\Gamma |- \gminv{e_1}{m}{U'}{e_2}: \stype{\ssubst{T}{U'}{X}}{\top}$ and 
			$\DeltaX;\Gamma |- \gminv{e'_1}{m}{U'}{e'_2}: \stype{\ssubst{T}{U'}{X}}{\top}$. It follows from above.
			\item $\sigma \in \nsetd{\DeltaX}, (k,\gamma_1,\gamma_2) \in \nsetg{\sigma(\Gamma)}$. It follows from above.
		\end{itemize}
		
		Hence, $(k, \sigma(\gamma_1(\gminv{e_1}{m}{U'}{e_2})), \sigma(\gamma_2(\gminv{e'_1}{m}{U'}{e'_2}))) \in \gsetc{\stype{\sigma(\ssubst{T}{U'}{X})}{\top}}$ 
		\end{enumerate}
\end{proof}

\subsubsection{Compatibility TPmD}
\begin{lemma}[\gobsec Compatibility TPmD]
\label{lm:gobsec-compatibility-primitive-method-decl-n}
\mbox{}
\\
Let $\DeltaX;\Gamma |- e_1 \approx e'_1 :\stype{T}{U}$ \\
Let $\DeltaX |- m \in U,~\tlookup{\DeltaX,U}{m}{\stype{\primt_1}{\ilab} -> \stype{\primt_2}{\ilab}}$	\\
Let $\DeltaX;\Gamma |- e_2 \approx e'_2 : \stype{\primt_1}{U_1}$ \\
Let $\rpolicy(\stype{\primt_1}{U_1},\primt_2) = \primt'_2$ \\
Then $\DeltaX;\Gamma |- \minv{e_1}{m}{e_2} \approx \minv{e'_1}{m}{e'_2}: \stype{\primt_2}{\primt'_2}$
\end{lemma}

\begin{proof}
Proof obligations:
\begin{enumerate}
	\item $\DeltaX;\Gamma |- \minv{e_1}{m}{e_2}: \stype{\primt_2}{\primt'_2}$ and $\DeltaX;\Gamma |- \minv{e'_1}{m}{e'_2}: \stype{\primt_2}{\primt'_2}$. Apply rule (TPmD).
	\item Assuming arbitrary $k, \sigma, \gamma_1, \gamma_2$ such as: 
		$k \geq 0,\sigma \in \nsetd{\DeltaX}, (k,\gamma_1,\gamma_2) \in \nsetg{\sigma(\Gamma)}$
		
		\textbf{Show:}
		\begin{center}
			$(k, \sigma(\gamma_1(\minv{e_1}{m}{e_2})), \sigma(\gamma_2(\minv{e'_1}{m}{e'_2}))) \in \nsetc{\sigma(\stype{\primt_2}{\primt'_2})}$ \\		
			$\equiv 
			(k, 
			\minv{\sigma(\gamma_1(e_1))}{m}{\sigma(\gamma_1(e_2))}, 
			\minv{\sigma(\gamma_2(e'_1))}{m}{\sigma(\gamma_2(e'_2))}) \in \nsetc{\sigma(\stype{\primt_2}{\primt'_2})}$\\
			$\equiv 
			(k, 
			\minv{\sigma(\gamma_1(e_1))}{m}{\sigma(\gamma_1(e_2))}, 
			\minv{\sigma(\gamma_2(e'_1))}{m}{\sigma(\gamma_2(e'_2))}) \in \nsetc{\Gbox{\stype{\primt_2}{\primt'_2}}}$
		\end{center}
		
		Instantiate the hypothesis $\DeltaX;\Gamma |- e_1 \approx e'_1 :\stype{T}{U}$ with $k, \sigma,\gamma_1,\gamma_2$, hence:
		
		$(k, \rhosyn{\gamma_1(e_1)},\rhosyn{\gamma_2(e'_1)}) \in \nsetc{\sigma(\stype{T}{U})}$
		
		Let $k' \le k$ and let $(k',v_1,v_2) \in \gsetv{\sigma(\stype{T}{U})}$. By Lemma~\ref{lm:gobsec-monadic-bind-n} (\nameref{lm:gobsec-monadic-bind-n}) we can rewrite the goal \textbf{to show}
		\begin{center}
		$(k', \minv{v_1}{m}{\rhosyn{\gamma_1(e_2)}}, \minv{v_2}{m}{\rhosyn{\gamma_2(e'_2)}}) \in \gsetc{\stype{\primt_2}{\primt'_2}}$ 
		\end{center}
		
		Instantiate the hypothesis $\DeltaX;\Gamma |- e_2 \approx e'_2 : \stype{\primt_1}{U_1}$ with $k, \rho,\gamma_1,\gamma_2$, hence: 
		
		$(k, \rhosyn{\gamma_1(e_2)},\rhosyn{\gamma_2(e'_2)} \in \gsetc{\sigma(\stype{\primt_1}{U_1})}$
		
		Let $k'' \le k'$ and let $(k'',v''_1,v''_2) \in \gsetv{\sigma(\stype{\primt_1}{U_1})}$. By Lemma~\ref{lm:gobsec-monadic-bind-n} (\nameref{lm:gobsec-monadic-bind-n}) we can rewrite the goal \textbf{to show}
		\begin{center}
		$(k'', \minv{v_1}{m}{v''_1}, \minv{v_2}{m}{v''_2})) \in \gsetc{\stype{\primt_2}{\primt'_2}}$ 
		\end{center}
		
		Instantiate $(k',v_1,v_2) \in \gsetv{\sigma(\stype{T}{U})}$ with $m, k'',v''_1,v''_2$. Note that:
		\begin{itemize}
			\item $m \in \sigma(U)$. It follows from $\DeltaX |- m \in U$, $\sigma \in \nsetd{\DeltaX}$ and $\wft{\DeltaX}{\stype{T}{U}}$. Then 
			 
			$\tlookup{\bigcdot,\sigma(T_1)}{m}{\stype{\primt_1}{\ilab} -> \stype{\primt_2}{\ilab}}$.
			\item $k'' < k$ which follows above assumptions.
			\item $(k'',v''_1,v''_2) \in \gsetv{\stype{\primt_1}{\sigma(U_1)}}$. Apply Lemma~\ref{lm:gobsec-monotonicity} (\nameref{lm:gobsec-monotonicity}) 
			with $(k,v''_1,v''_2) \in \gsetv{\stype{\primt_1}{\sigma(U_1)}}$  and $k'' \le k$ 
		\end{itemize}

		Hence, $(k'',\minv{v_1}{m}{v''_1},\minv{v_2}{m}{v''_2}) \in \gsetc{\stype{\primt_2}{\rpolicy(\stype{\primt_1}{\rho(U_1)},\primt_2)}}$
		
		Then, note that $\rpolicy(\stype{\primt_1}{\rho(U_1)},\primt_2) = \rpolicy(\stype{\primt_1}{U_1},\primt_2) = \primt'_2$.
		
		Hence, $(k'', \minv{v_1}{m}{v''_1}, \minv{v_2}{m}{v''_2})) \in \gsetc{\stype{\primt_2}{\primt'_2}}$ 
\end{enumerate}
\end{proof}

\subsubsection{Compatibility TPmH}
\begin{lemma}[\gobsec Compatibility TPmH]
\label{lm:gobsec-compatibility-primitive-method-high-n}
\mbox{}
\\
Let $\DeltaX;\Gamma |- e_1 \approx e'_1 :\stype{T}{U}$ \\
Let $\DeltaX |- m \in T,~\tlookup{\DeltaX,T}{m}{\stype{\primt_1}{\ilab} -> \stype{\primt_2}{\ilab}}$	\\
Let $\DeltaX;\Gamma |- e_2 \approx e'_2 : \stype{\primt_1}{U_1}$ \\
Then $\DeltaX;\Gamma |- \minv{e_1}{m}{e_2} \approx \minv{e'_1}{m}{e'_2}: \stype{\primt_2}{\top}$
\end{lemma}

\begin{proof}
Proof obligations:
\begin{enumerate}
	\item $\DeltaX;\Gamma |- \minv{e_1}{m}{e_2}: \stype{\primt_2}{\top}$ and 
	$\DeltaX;\Gamma |- \minv{e'_1}{m}{e'_2}: \stype{\primt_2}{\top}$. Apply rule (TPmH).
	\item Assuming arbitrary $k, \sigma, \gamma_1, \gamma_2$ such as: 
		$k \geq 0,\sigma \in \nsetd{\DeltaX}, (k,\gamma_1,\gamma_2) \in \nsetg{\sigma(\Gamma)}$
		
		\textbf{Show:}
		\begin{center}
			$(k, \sigma(\gamma_1(\minv{e_1}{m}{e_2})), \sigma(\gamma_2(\minv{e'_1}{m}{e'_2}))) \in \gsetc{\stype{\sigma(\primt_2)}{\top}}$ 
		\end{center}
		
		Apply Lemma~\ref{lm:gobsec-well-typed-term-related-top} (\nameref{lm:gobsec-well-typed-term-related-top}) with:
		\begin{itemize}
			\item $\DeltaX;\Gamma |- \minv{e_1}{m}{e_2}: \stype{\primt_2}{\top}$ and 
	$\DeltaX;\Gamma |- \minv{e'_1}{m}{e'_2}: \stype{\primt_2}{\top}$. It follows from above
			\item $\sigma \in \nsetd{\DeltaX}, (k,\gamma_1,\gamma_2) \in \nsetg{\sigma(\Gamma)}$. It follows from above.
		\end{itemize}
		
		Hence, $(k, \sigma(\gamma_1(\minv{e_1}{m}{e_2})), \sigma(\gamma_2(\minv{e'_1}{m}{e'_2}))) \in \gsetc{\stype{\sigma(\primt_2)}{\top}}$ 
\end{enumerate}
\end{proof}

\subsubsection{Fundamental property}
\gfprni*
\begin{proof}
The proof is by induction on the typing derivation of $\DeltaX;\Gamma |- e : S$.

Each case follows directly from the corresponding compatibility lemma 
(Lemmas~\ref{lm:gobsec-compatibility-var-n} ... \ref{lm:gobsec-compatibility-primitive-method-high-n})
\end{proof}

\fpImpliesPRNI*
\begin{proof}
The proof is direct given the similarity of both definitions. The only difference between both definitions is that 
${\Delta,\Gamma |- e \approx e:S}$
uses the security type system (Figure~\ref{fig:gobsec-static-semantics}) and $\gtrni{\Delta}{\Gamma}{e}{S}$ uses the
simple type system (Figure~\ref{fig:gobsec-safe-type-system}). We use Lemma~\ref{lm:gobsec-security-ts-implies-simple-ts} (\nameref{lm:gobsec-security-ts-implies-simple-ts})
to show that $\DeltaX;\Gamma |- e: \stype{T}{U} => \stypeof{\Gamma}{e}{T}$
\end{proof}

\renewcommand{\rhosyn}[1]{\rho(#1)} 
\renewcommand{\rhosynx}[2]{{#1}_\mathsf{syn}(#2)}
\renewcommand{\gsetv}[1]{\gsetvx{#1}{\rho}}
\renewcommand{\gsetc}[1]{\gsetcx{#1}{\rho}}
\renewcommand{\gsetd}[1]{\mathcal{D}\llbracket#1\rrbracket}
\renewcommand{\gsetg}[1]{\gsetgx{#1}{\rho}}
\label{sec:gobsec-appendix-prni-proof}

\clearpage
\section{List of figures}
\listoffigures

\section{List of theorems}
\listoftheorems[ignoreall,show={lemma,theorem}]

\fi

\end{document}